\documentclass[a4paper]{amsart}

\usepackage[utf8]{inputenc}
\usepackage[authoryear,round]{natbib}
\usepackage{amsmath}
\usepackage{mathtools}
\usepackage{color}
\usepackage{enumerate}
\usepackage{algorithm}
\usepackage[commentColor=black,beginComment=,italicComments=false]{algpseudocodex} 
\usepackage[colorlinks=true,linkcolor=blue,citecolor=blue]{hyperref}

\newcommand{\E}{\mathbb{E}}    
\renewcommand{\P}{\mathbb{P}}  
\newcommand{\R}{\mathbb{R}}    
\newcommand{\I}{\mathbb{I}} 

\newcommand{\uarg}{\,\cdot\,}  
\newcommand{\M}{\mathbf{M}}    
\newcommand{\osc}{\operatorname{osc}} 
\newcommand{\ud}{\mathrm{d}}   
\renewcommand{\vec}[1]{\mathbf{#1}} 
\newcommand{\given}{\,:\,}    

\newcommand{\G}{\mathcal{G}} 
\newcommand{\X}{\mathsf{X}} 
\newcommand{\tv}[1]{\left \| #1 \right \|_{\rm TV}}
\newcommand{\HD}{\mathcal{H}} 

\newcommand{\fsp}{\zeta} 
\newcommand{\proj}{\mathrm{proj}}

\newdimen\slantmathcorr
\def\oversl#1{
\setbox0=\hbox{$#1$}
\slantmathcorr=\wd0
\hskip 0.2\slantmathcorr \overline{\hbox to 0.8\wd0{%
\vphantom{\hbox{$#1$}}}}
\hskip-\wd0\hbox{$#1$}
}

\def\undersl#1{
\setbox0=\hbox{$#1$}
\slantmathcorr=\wd0
\underline{\hbox to 0.8\wd0{%
\vphantom{\hbox{$#1$}}}}
\hskip-0.8\wd0\hbox{$#1$}
}

\newcommand{\Mlo}{\undersl{M}}
\newcommand{\Glo}{\undersl{G}}
\newcommand{\Mhi}{\oversl{M}}
\newcommand{\Ghi}{\oversl{G}}

\newtheorem{theorem}{Theorem}
\newtheorem{corollary}[theorem]{Corollary}
\newtheorem{proposition}[theorem]{Proposition}
\newtheorem{lemma}[theorem]{Lemma}

\newtheorem{assumption}{Assumption}

\newtheorem*{assumption*}{Assumption}
\newtheorem*{definition*}{Definition}

\theoremstyle{definition}

\theoremstyle{remark}
\newtheorem{remark}[theorem]{Remark}

\title{Mixing time of the conditional backward sampling particle filter}

\author{Joona Karjalainen}
\address{Joona Karjalainen, Department of Mathematics and Statistics, University of
Jyväskylä P.O.Box 35, FI-40014 University of Jyväskylä, Finland}
\author{Anthony Lee}
\address{Anthony Lee, School of Mathematics, University of Bristol,
  Woodland Rd, Bristol BS8 1UG, United Kingdom}
\author{Sumeetpal S. Singh}
  \address{Sumeetpal S.~Singh, 
NIASRA, School of Mathematics and Applied Statistics, University of Wollongong, NSW 2522, Australia}
\author{Matti Vihola}
\address{Matti Vihola, Department of Mathematics and Statistics, University of
Jyväskylä P.O.Box 35, FI-40014 University of Jyväskylä, Finland}



\begin{document}

\begin{abstract}
  The conditional backward sampling particle filter (CBPF) is a powerful Markov chain Monte Carlo sampler for general state space hidden Markov model (HMM) smoothing. It was proposed as an improvement over the conditional particle filter (CPF), which has an $O(T^2)$ complexity under a general `strong' mixing assumption, where $T$ is the time horizon. Empirical evidence of the superiority of the CBPF over the CPF has never been theoretically quantified. We show that the CBPF has $O(T \log T)$ time complexity under strong mixing: its mixing time is upper bounded by $O(\log T)$, for any sufficiently large number of particles $N$ independent of $T$. This $O(\log T)$ mixing time is optimal. To prove our main result, we introduce a novel coupling of two CBPFs, which employs a maximal coupling of two particle systems at each time instant.  The coupling is implementable and we use it to construct unbiased, finite variance, estimates of functionals which have arbitrary dependence on the latent state's path, with a total expected cost of $O(T \log T)$. We use this to construct unbiased estimates of the HMM's score function, 
  and also investigate other couplings which can exhibit improved behaviour. We demonstrate our methods on financial and calcium imaging applications. 
\end{abstract}

\maketitle

\section{Introduction}

General hidden Markov model (HMM; also known as non-linear/non-Gaussian state space model) smoothing is a common inference task in time-series analysis and in various engineering and machine learning applications \citep[e.g.][]{durbin-koopman,chopin-papaspiliopoulos,cappe-moulines-ryden}. The conditional particle filter (CPF) introduced in the seminal paper of \cite{andrieu-doucet-holenstein} defines a Markov chain Monte Carlo (MCMC) update, or transition, that targets the smoothing distribution. The CPF algorithm is similar to the particle filter \citep{gordon-salmond-smith}, but in the CPF, the `particles' are auxiliary variables within one MCMC update. The \emph{mixing time} of this MCMC transition, that is, the number of updates needed to produce a sample that is $\epsilon$-close in distribution to the target, can be improved arbitrarily by increasing the number of particles $N$ \citep{andrieu-lee-vihola,lindsten-douc-moulines}. However, this increases the cost: one iteration of the CPF has $O(NT)$ computational complexity, where $T$ is the time horizon or length of the data record.

In a discussion of the original paper, \cite{whiteley-backwards-note} suggested that backward sampling can be used within the CPF while keeping the transition in detailed balance with the smoothing distribution. This relatively small algorithmic modification of the CPF, which we call the conditional backward sampling particle filter (CBPF, Algorithm \ref{alg:cbpf}), has a dramatically improved performance over the CPF. In particular, the CBPF update is known to remain efficient with an increasing time horizon $T$ with a fixed number of particles $N$ \citep{lee-singh-vihola}, in contrast with the CPF which generally requires $N = O(T)$ particles \citep{andrieu-lee-vihola,lindsten-douc-moulines} to remain efficient, ensuring $O(1)$ mixing time. The particle Gibbs with ancestor sampling algorithm \citep{lindsten-jordan-schon} is probabilistically equivalent to the CBPF for models of the form \eqref{eq:feynman-kac-pi}, and therefore enjoys the same performance benefits as the CBPF does.

The CBPF mixing time upper bound established in \citep{lee-singh-vihola} is $O(T)$, which implies an overall computational cost upper bound $O(T^2)$ to produce a sample that is $\epsilon$-close in distribution to the target since the cost of running the CBPF is $O(NT)$. This is the same order of cost as the CPF using $N=O(T)$ particles, which is the number of particles needed for the CPF to remain effective as $T$ increases \citep{andrieu-lee-vihola,lindsten-douc-moulines}.  Thus, the result in  \citep{lee-singh-vihola} does not fully capture the numerous empirical findings that the CBPF can be far more efficient than the CPF \citep[e.g.][]{lindsten-schon,lindsten-jordan-schon,chopin-papaspiliopoulos}. The main result of this paper is a substantially improved $O(\log T)$ mixing time upper bound for the CBPF (Theorem \ref{thm:cbpf-mixing} in Section \ref{sec:cbpf}), under a (strong mixing) condition similar to \citep{lee-singh-vihola}. Hence, we find that the computational complexity is $O(T \log T)$, a significant improvement over $O(T^2)$. Furthermore, the $O(\log T)$ mixing time is tight, as illustrated by Proposition \ref{prop:logTexample}. 

We prove the mixing time result by analysing a novel and implementable coupling of two CBPFs (Algorithm \ref{alg:mc_cbpf} in Section \ref{sec:coupling-cbpf}), whose meeting time has an $O(\log T)$ upper bound (Corollary \ref{cor:coupling-time} of Theorem \ref{thm:main-contraction}). Because the coupling (Algorithm \ref{alg:mc_cbpf}) can be implemented in practice, it can be used for unbiased estimation of smoothing expectations, which we discuss in Section \ref{sec:unbiased}, along with other coupling variants. The methods we investigate complement and improve upon earlier couplings proposed for this same estimation task \citep{jacob-lindsten-schon,lee-singh-vihola}. 
Our $O(\log T)$ coupling time bound (Corollary \ref{cor:coupling-time}) ensures that unbiased, finite variance estimators can also be produced with $O(T\log T)$ cost. 
We illustrate the behaviour of the coupling algorithms with four HMMs in Section \ref{sec:experiments}. In Section \ref{sec:mle}, we construct unbiased estimates of the HMM's score function, leading to stochastic gradient maximum likelihood estimation for the stochastic volatility model with leverage. Another real-data application in Section \ref{sec:fluorescence} is about calcium fluorescence imaging for inferring neural spike trains. The paper then concludes with a discussion in Section \ref{sec:discussion}. 

Finally, Appendix \ref{app:cpf-theory} and Appendix \ref{app:contraction} are devoted to intermediate results which lead to the proof of Theorem \ref{thm:main-contraction} on the meeting time of our novel coupling of CBPFs. Appendix \ref{app:cpf-theory} contains important (but technical) properties of the CPF. Appendix \ref{app:contraction} contains the main argument of the proof, which captures the `progressive' nature of the coupling: the expected number of uncoupled states decreases geometrically over iterations, as long as the number of particles is large enough (with respect to the strong mixing constants).

\section{Conditional backward sampling particle filter}
\label{sec:cbpf}

Throughout the paper, we use the notation $a{:}b$ for integers from $a$ to $b$ (inclusive), and denote in short $x_{a:b} = (x_a,\ldots,x_b)$ and $x^{a:b} = (x^a,\ldots,x^b)$. A hidden Markov model (HMM) consists of a latent (unobserved) Markov chain $X_{1:T}$ where each $X_t$ takes values on a general state space $\X$, and observable random variables $Y_{1:T}$ which are conditionally independent given $X_{1:T}$. The model is defined in terms of the initial density\footnote{While our results are stated for densities, they hold also for the discrete case with probability masses, and more generally: the densities can be with respect to any $\sigma$-finite dominating measure.} $X_1\sim M_1(\uarg)$ and the Markov transition densities $X_t \sim M_t(X_{t-1},\uarg)$ for $t=2,\ldots,T$. The observations $Y_t\mid (X_t=x_t)$ have conditional densities $g_t(x_t,y_t)$.
From hereon, we assume the observations $y_{1:T}$ are fixed, and define $G_t(x_t) = g_t(x_t, y_t)$. This is the Feynman--Kac representation of the HMM \citep{del2004feynman}, which abstracts $(G_t)$ to be any sequence of non-negative functions as long as the probability distribution $\pi$ with density (also denoted by $\pi$):
\begin{equation}
   \pi(x_{1:T}) \propto M_1(x_1) G_1(x_1)\prod_{t=2}^T M_t(x_{t-1},x_t) G_t(x_t)
   \label{eq:feynman-kac-pi}
\end{equation}
is well-defined. In the HMM setting, $\pi(x_{1:T}) = p(x_{1:T}\mid y_{1:T})$, which is commonly known as the smoothing distribution associated with the HMM. However, models of the same form \eqref{eq:feynman-kac-pi} also emerge beyond the HMM setting, for instance, in the context of time-discretised path integral models \citep[cf.][]{karppinen-singh-vihola}.

The CBPF defines a Markov transition $\X^T \to \X^T$ which has $\pi$ as its stationary distribution.
\begin{algorithm}
  \caption{\textsc{CBPF}($x_{1:T}^*,N$).}
  \label{alg:cbpf}
  \begin{algorithmic}[1]
      \State $X_{1:T}^{0} \gets x_{1:T}^*$ \label{line:cbpf-set-ref}
      \State $X_1^{i}\sim M_1(\uarg)$ \label{line:cbpf-forward-start} \Comment{for $i\in\{1{:}N\}$}
      \State $W_1^{i} \gets G_1(X_1^{i})$ \Comment{for $i\in\{0{:}N\}$}
      \For{$t=2,\ldots,T$}
      \State 
      $X_t^{i} \sim 
      \sum_{k=0}^N \frac{W_{t-1}^k}{\sum_{j=0}^N W_{t-1}^j} M_t(X_{t-1}^k, \uarg)$ \label{line:mixture-transition} \Comment{for $i\in\{1{:}N\}$}
      \State $W_t^{i} \gets G_t(X_t^{i})$  
      \label{line:cbpf-forward-end}
      \Comment{for $i\in\{0{:}N\}$} 
      \EndFor
      \State $J_T  \sim \mathrm{Categorical}(W_T^{0:N})$
      \For{$t=T-1,T-2,\ldots,1$}
      \State $B_t^{i} \gets W_t^{i} M_{t+1}(X_t^{i}, X_{t+1}^{J_{t+1}})$ \label{line:cbpf-backward-sample1} \Comment{for $i\in\{0{:}N\}$}
      \State $J_t\sim \mathrm{Categorical}(B_t^{0:N})$ \label{line:cbpf-backward-sample2} 
      \EndFor
      \State \textbf{output}
      $(X_1^{J_1},\ldots,X_T^{J_T})$
  \end{algorithmic}
\end{algorithm}
In Algorithm \ref{alg:cbpf}, and throughout this work, $\mathrm{Categorical}(\omega^{0:N})$ stands for the categorical distribution with probabilities \emph{proportional to} the weights $\omega^{0:N}\ge 0$, that is, $I \sim \mathrm{Categorical}(\omega^{0:N})$ satisfy $\P(I=i) = \omega^i/\sum_{j=0}^N \omega^j$ for $i\in\{0{:}N\}$. In algorithmic descriptions, the symbol $\sim$ means that the random variable on the left follows the conditional distribution defined on the right, and is conditionally independent of other random variables generated earlier.

Our theoretical results for Algorithm \ref{alg:cbpf} are stated under the following `strong mixing' assumption.
\begin{assumption}
  \label{a:strong-mixing}
  There exist constants $0 < \Mlo \le \Mhi <\infty$ and
  $0 < \Glo \le \Ghi <\infty$ such that for all $x,x' \in \X$:
  \begin{itemize}
    \item $(M)$: $\Mlo \le M_1(x)\le \Mhi$ and $\Mlo \le M_t(x,x') \le \Mhi$ for $t=2,\ldots,T$;
    \item $(G)$: $\Glo \le G_{t}(x) \le \Ghi$ for $t=1,\ldots,T$.
  \end{itemize}
\end{assumption}
This assumption, or closely related variants, is common in the theory of particle filters \citep[e.g.][]{del2004feynman} as well as earlier results for the CBPF \citep{chopin-singh,lee-singh-vihola}. In many applications, it can be satisfied when we restrict the model to a compact $\X\subset\R^d$ without materially affecting the inference.

Hereafter, we denote the total variation distance between two probability measures $P$ and $Q$ as $\tv{P-Q} = \sup_A | P(A) - Q(A) |$, where the supremum is over all events $A$.
For the remainder of this section, let $N\ge 1$ be fixed, let $\mathbf{S}_0 = x_{1:T}^* \in \X^T$ be an arbitrary initial path, and define the iterates of the CBPF Markov transition as
$\mathbf{S}_k \gets \textsc{CBPF}(\mathbf{S}_{k-1}, N)$ for $k=1,2,\ldots$, and denote by $\mathrm{Law}_N(\mathbf{S}_k)$ the distribution of $\mathbf{S}_k$. 

\begin{theorem}
\label{thm:cbpf-mixing}
There exist finite constants $N_{\rm min}$ and $c_r$ which depend only on the constants in \textup{\ref{a:strong-mixing}}, such that for all $N\ge N_{\rm min}$, any $T\ge 1$, any initial state $x_{1:T}^*\in\X^T$ and $k\ge 1$, the following upper bound for the total variation distance holds:
\begin{equation}
\| \mathrm{Law}_N(\mathbf{S}_k) - \pi \|_\mathrm{TV} \le T r_N^{-k},
\qquad\text{where}\qquad
   r_N \ge c_r \frac{\sqrt{N}}{\log^2(N)}.
   \label{eq:r-rate}
\end{equation}
\end{theorem}

Theorem \ref{thm:cbpf-mixing}, whose proof is given at the end of Section \ref{sec:coupling-cbpf}, guarantees that whenever $N$ is large enough so that 
$r_N>1$, the following $O(\log T)$ mixing time upper bound holds for the CBPF Markov transition:
$$
  t_\mathrm{mix}(\epsilon) = \inf\{k\ge 0\given \sup_{x_{1:T}^*}\| \mathrm{Law}_N(\mathbf{S}_k) - \pi \|_\mathrm{TV} \le \epsilon \}
  \le  \Big\lceil \frac{\log(T) - \log(\epsilon)}{\log(r_N)}\Big\rceil.
$$

The Markov chain $(\mathbf{S}_k)$ is known to be reversible for any $N$, following \cite{chopin-singh}, and so it can be natural to consider convergence in $\mathrm{L}^2(\pi)$ and the spectrum of its corresponding Markov operator. For a signed measure $\nu \ll \pi$, we denote
$
\| \nu \|_{\mathrm{L}^2(\pi)}^2 = \int \left |\mathrm{d} \nu/\mathrm{d} \pi \right |^2 \mathrm{d} \pi.
$
We note that for a $\pi$-reversible Markov kernel $P$, the convergence of $\| \nu P^k \|_{\mathrm{L}^2(\pi)}$ is equivalent to the convergence of $\| P^k f \|_{\mathrm{L}^2(\pi)}$ where $f = \frac{\mathrm{d} \nu}{\mathrm{d} \pi}$.
\begin{proposition}
Let $\mathrm{Law}_N(\mathbf{S}_0) \ll \pi$. Then the $\mathrm{L}^2$-rate of convergence satisfies
$$
\| \mathrm{Law}_N(\mathbf{S}_k) - \pi \|_{\mathrm{L}^2(\pi)} \leq \| \mathrm{Law}_N(\mathbf{S}_0) - \pi \|_{\mathrm{L}^2(\pi)} r_N^{-k},
$$ 
and the spectrum of the corresponding Markov operator acting on $\mathrm{L}_0^2(\pi) = \{f \in \mathrm{L}^2(\pi) : \pi(f) = 0 \}$ is contained in $[0,r_N^{-1}]$.
\end{proposition}

\begin{proof}
The CBPF Markov chain is reversible by \cite[][Proposition~9]{chopin-singh}. The corresponding operator $P_N$ is positive by \cite[][Theorem~10]{chopin-singh}.
From Theorem \ref{thm:cbpf-mixing}, reversibility and \cite[][Theorem~2.1]{roberts2007geometric} we may deduce that the spectral radius of $P_N$ (acting on $\mathrm{L}_0^2(\pi)$) is less than or equal to $r_N^{-1}$ and the bound on the $\mathrm{L}^2(\pi)$ norm. The conclusion about the spectrum follows from the positivity of $P_N$.
\end{proof}

The following simple example shows that the $O(\log T)$ mixing time above cannot be improved under \ref{a:strong-mixing}.

\begin{proposition}
  \label{prop:logTexample}
Assume that $\mathsf{X} = [0,1]$ and $M_1(\uarg) = M_t(x,\uarg)$ are uniform $[0,1]$ distributions, and let $G_t \equiv 1$. Then, for any $x_{1:T}^*\in \mathsf{X}^T$, any fixed $N\ge1$ and any sequence $k = k(T) = o(\log T)$:
$$
\liminf_{T\to\infty}
   \tv{ \mathrm{Law}_N(\mathbf{S}_k) - \pi }
   = 1.
$$
\end{proposition}
\begin{proof}
Note that for this specific choice of $M_t$ and $G_t$, the particles $X_t^{i}$ in Algorithm \ref{alg:cbpf} are all independent $U([0,1])$ random variables, and $J_t$ are independent $U(\{0{:}N\})$ random variables, so $X_t^{J_t}=x_t^\ast$ with probability $1/(N+1)$. 
Denote by $X_t(k) = [\mathbf{S}_k]_t$ the state at time $t$ of the $k$th CBPF iterate, and let 
$\tau_t = \inf\{k\given X_t(k) \neq x_t^* \}$ be the iteration when the state at time $t$ has moved the first time.
Clearly, $\tau_t \sim \mathrm{Geom}(\frac{N}{N+1})$ and $\P(\tau_t \le k) = 1 - (N+1)^{-k}$.

Let $A = \{x_{1:T}\in [0,1]^T \given x_1 \neq x_1^*, \ldots, x_T \neq x_T^*\}$, note that $\pi$ is uniform on the unit cube $[0,1]^T$ so $\pi(A)=1$, and
$$
\P(\mathbf{S}_k \in A)= \P(\tau \le k),
$$
where $\tau = \max\{\tau_1,\ldots,\tau_T\}$. Because $\tau_i$ are independent,
\begin{equation}
\P(\tau \le k) = \P(\tau_1\le k)^T = \big(1 - (N+1)^{-k}\big)^T.
\label{eq:example-tau}
\end{equation}
If $k = o(\log T)$, that is $k = \epsilon(T) \log(T)$ where $\epsilon(T) = o(1)$, 
then it is direct to verify that $(1 - (N+1)^{-k})^T
= (1 - T^{-\epsilon(T)\log(N+1)})^T \to 0$, and consequently,
$
\| \mathrm{Law}_N(\mathbf{S}_k)  - \pi \|_{\mathrm{TV}} 
\ge |\P(\mathbf{S}_k \in A) - 1 | \to 1.
$ as $T\to\infty$.
\end{proof}

\begin{remark}
  \label{rem:logTexample-Nrate}
  In fact, in the proof of Proposition \ref{prop:logTexample}, $\mathbf{S}_\tau \sim \pi$, so the model satisfies
  $$
  \tv{\mathrm{Law}_N(\mathbf{S}_k)  - \pi }
  = \tv{\mathrm{Law}_N(\mathbf{S}_k) -  \mathrm{Law}_N(\mathbf{S}_\tau)  } \le \P(\tau > k) \le T (N+1)^{-k},
  $$ 
  where the first inequality follows from the coupling inequality, and the last inequality follows from \eqref{eq:example-tau} and Bernoulli's inequality. Therefore, this example satisfies the same form of upper bound as in Theorem \ref{thm:cbpf-mixing}, with $r_N = N+1$. This dependence on $N$ is more favourable than in our bound \eqref{eq:r-rate}, and is achieved by taking advantage of the independence in the model, i.e. $M_t(x,x')=M_t(x')$. In general, the particle system's behaviour is more complicated, and may indeed lead to a different rate in terms of $N$. 
  It is therefore unclear how tight the rate for $r_N$ given in \eqref{eq:r-rate} is in terms of $N$.
\end{remark}

We conclude this section by comparing Theorem \ref{thm:cbpf-mixing} with earlier mixing bounds for the CBPF. 
The quantitative bound established for the CPF in \citep{andrieu-lee-vihola} is also valid for the CBPF. Indeed, \citep[][Lemma 16]{chopin-singh-arxiv-v1} together with 
\citep[Corollary 14]{andrieu-lee-vihola} (see Appendix \ref{app:cpf-cbpf} for more details) imply that there exists a constant $C\in(0,\infty)$ depending only on \ref{a:strong-mixing} such that:
\begin{equation}
\tv{\mathrm{Law}_N(\mathbf{S}_k)  - \pi }
\le \big(1 - e^{-TC/N}\big)^k.
\label{eq:andrieu-lee-vihola-bound}
\end{equation}
This bound is appropriate for small $T$ and large $N$; then $1 - e^{-T C/N} = O(N^{-1})$, which coincides with the observation in Remark \ref{rem:logTexample-Nrate}. Note, however, that we must take $N = O(T)$ for \eqref{eq:andrieu-lee-vihola-bound} to be useful; for fixed $N$, a mixing time bound from \eqref{eq:andrieu-lee-vihola-bound} increases exponentially in $T$.

The only earlier mixing time result for the CBPF, which remains useful for fixed $N$ and large $T$, is Theorem 4 of \citep{lee-singh-vihola}: for any $\alpha,r>1$, there exists $N_0<\infty$ depending on 
$\alpha$, $r$ and \ref{a:strong-mixing} such that for all $N\ge N_0$:
\begin{equation}
\tv{\mathrm{Law}_N(\mathbf{S}_k)  - \pi }
\le \alpha^T r^{-k}.
\label{eq:lee-singh-vihola-bound}
\end{equation}
This bound is similar to Theorem \ref{thm:cbpf-mixing}, but with $T$ replaced by a much more pessimistic $\alpha^T$, only guaranteeing $O(T)$ mixing time.

\section{Coupling of CBPF transitions}
\label{sec:coupling-cbpf}

The proof of Theorem \ref{thm:cbpf-mixing} relies on the analysis of a coupled CBPF (Algorithm \ref{alg:mc_cbpf}), which simulates two CBPF transitions with references $x_{1:T}^*$ and $\tilde{x}_{1:T}^*$, respectively.
\begin{algorithm}
    \caption{\textsc{CoupledCBPF}($x_{1:T}^*,\tilde{x}_{1:T}^*,N$).}
    \label{alg:mc_cbpf}
    \begin{algorithmic}[1]
        \State $\big(X_{1:T}^{0}, \tilde{X}_{1:T}^{0}\big)
        \gets \big(x_{1:T}^*, \tilde{x}_{1:T}^*\big)$.
        \State $X_1^{i}\gets \tilde{X}_1^{i} \sim M_1(\uarg)$ \label{line:initial-coupling} \label{line:forward-start}        \Comment{for $i\in\{1{:}N\}$}
        \State $W_1^{i} \gets G_1(X_1^{i})$;
        $\tilde{W}_1^{i} \gets G_1(\tilde{X}_1^{i})$ \Comment{for $i\in\{0{:}N\}$}
        \For{$t=2,\ldots,T$}
        \State  $(X_t^{1:N}, \tilde{X}_t^{1:N}) \sim \textsc{FwdCouple}(X_{t-1}^{0:N}, W_{t-1}^{0:N},
        \tilde{X}_{t-1}^{0:N}, \tilde{W}_{t-1}^{0:N}, M_t, t)$ \label{line:maximalcoupling}
         \State $W_t^{i} \gets G_t(X_t^{i})$; $\tilde{W}_t^{i} \gets        G_t(
\tilde{X}_t^{i})$ \Comment{for $i\in\{0{:}N\}$}
\label{line:coupled-weights}
        \EndFor \label{line:forward-end}
        \State $(J_T,\tilde{J}_T) \sim
        \textsc{MaxCouple}\big(\mathrm{Categorical}(W_T^{0:N}),
        \mathrm{Categorical}(\tilde{W}_T^{0:N}) \big)$ \label{line:backward-start}
        \For{$t=T-1,T-2,\ldots,1$}
        \State $B_t^{i} \gets W_t^{i} M_{t+1}(X_t^{i},
        X_{t+1}^{J_{t+1}})$; 
        $\tilde{B}_t^{i} \gets \tilde{W}_t^{i}
        M_{t+1}(\tilde{X}_t^{i},\tilde{X}_{t+1}^{\tilde{J}_{t+1}})$ \Comment{for $i\in\{0{:}N\}$}
        \label{line:coupled-backward-weights}
        \State $(J_t, \tilde{J}_t) \sim
        \textsc{MaxCouple}\big(\mathrm{Categorical}(B_t^{0:N}),
        \mathrm{Categorical}(\tilde{B}_t^{0:N}) \big)$
        \EndFor \label{line:backward-end}
        \State \textbf{output}
        $\big((X_1^{J_1}, \ldots, X_T^{J_T}), (\tilde{X}_1^{\tilde{J}_1}, \ldots, \tilde{X}_T^{\tilde{J}_T})\big)$
    \end{algorithmic}
\end{algorithm}

In Algorithm \ref{alg:mc_cbpf}, $\textsc{MaxCouple}(\mu,\nu)$ stands for any maximal coupling of distributions $(\mu,\nu)$. If $(X,Y)\sim \textsc{MaxCouple}(\mu,\nu)$, then $\P(X\neq Y) = \| \mu - \nu \|_\mathrm{TV}$ \citep{thorisson}; see Appendix \ref{app:maxcouple} for algorithms that simulate from maximal couplings. 

Algorithm \ref{alg:mc_cbpf} is identical to the coupled CBPF in \citep{lee-singh-vihola}, except for \textsc{FwdCouple}, which accommodates more general coupling strategies for the forward pass of the algorithm.
In particular, let $\fsp_t^N$ and $\tilde{\fsp}_t^N$ denote the one-step predictive distributions of the two CPFs (cf. line \ref{line:mixture-transition} of Algorithm \ref{alg:cbpf}):
\begin{equation}
   \fsp_t^N = \sum_{i=0}^N \frac{W_{t-1}^i}{\sum_{j=0}^N W_{t-1}^j} M_t(X_{t-1}^i, \uarg),
   \qquad
   \tilde{\fsp}_t^N
   = \sum_{i=0}^N \frac{\tilde{W}_{t-1}^i}{\sum_{j=0}^N \tilde{W}_{t-1}^j} 
   M_t(\tilde{X}_{t-1}^i, \uarg).
   \label{eq:filter-state-predictive}
\end{equation}
\textsc{FwdCouple} may be chosen to implement any coupling of $(\fsp_t^N)^{\otimes N}$ and $(\tilde{\fsp}_t^N)^{\otimes N}$ (the $N$-fold product distributions of the above), and this choice can vary with $t$.

Because the forward pass simulates from couplings of the predictive distributions $\fsp_t^N$ and $\tilde{\fsp}_t^N$, it is direct to verify that 
$(X_{1:T}^{1:N})$ and $(\tilde{X}_{1:T}^{1:N})$ have the same law as Algorithm \ref{alg:cbpf}'s forward pass' outputs, after being invoked twice with inputs $x_{1:T}^*$ and $\tilde{x}_{1:T}^*$, respectively.

Finally,  Algorithm \ref{alg:mc_cbpf} uses $\textsc{MaxCouple}$ in the backward pass. Thus, upon completion, it indeed simulates from the coupling of the two CBPFs. That is, if
$$
  (X_{1:T},\tilde{X}_{1:T}) \gets \textsc{CoupledCBPF}(x_{1:T}^*, \tilde{x}_{1:T}^*, N),
$$
then the probability distributions of its marginals are equivalent to:
$$
X_{1:T} \gets \textsc{CBPF}(x_{1:T}^*, N)\qquad \text{and}\qquad  
\tilde{X}_{1:T} \gets \textsc{CBPF}( \tilde{x}_{1:T}^*, N).
$$

In all our theoretical analysis, including Theorem \ref{thm:main-contraction} and Corollary \ref{cor:coupling-time} below, which lead to the proof of Theorem \ref{thm:cbpf-mixing}, the following implementation of the forward pass steps will be assumed:
\begin{assumption}
  \label{a:joint-maximal-coupling}
$\textsc{FwdCouple}(X_{t-1}^{0:N}, W_{t-1}^{0:N}, \tilde{X}_{t-1}^{0:N}, \tilde{W}_{t-1}^{0:N}, M_t, t)$ is equivalent to
\begin{equation*}
\textsc{MaxCouple}\big((\fsp_t^N)^{\otimes N}, (\tilde{\fsp}_t^N)^{\otimes N}\big). 
\end{equation*}
\end{assumption}
Note that if $x_{1:T}^* = \tilde{x}_{1:T}^*$, then under 
Assumption \ref{a:joint-maximal-coupling} all variables in Algorithm \ref{alg:mc_cbpf} are coupled: $X_t^i = \tilde{X}_t^i$ (almost surely), and therefore also the outputs coincide: $X_t^{J_t}=\tilde{X}_t^{\tilde{J}_t}$ for $t\in\{1{:}T\}$.

In spite of the restriction to Assumption \ref{a:joint-maximal-coupling} in the theoretical analysis, we retain the full generality of $\textsc{FwdCouple}$ in Algorithm \ref{alg:mc_cbpf}. This is because other couplings can be interesting when the algorithm is implemented, and used for constructing unbiased estimates of smoothing functionals. 
We return to discuss other coupling strategies and unbiased estimation in Section \ref{sec:unbiased}. 

The CBPF mixing time bound in Theorem \ref{thm:cbpf-mixing} is based on the following main technical result, which establishes a contraction for the expected number of uncoupled states after just one application of the coupled CBPF algorithm. 
This result formalises the fact that for sufficiently large $N$ there are spontaneous partial couplings of particles across the trajectory. This arises because the distributions in \eqref{eq:filter-state-predictive} can be close even if all of the particles $X_{t-1}^i$ and $\tilde{X}_{t-1}^i$ are different, and therefore under \ref{a:joint-maximal-coupling} the particles $X_t^i$ and $\tilde{X}_t^i$ can be equal. This is in contrast with the progressive behaviour of index couplings seen in \citep{lee-singh-vihola}.

\begin{theorem}
\label{thm:main-contraction}
For any $(x_{1:T}^*,\tilde{x}_{1:T}^*)\in \mathsf{X}^{T}\times\mathsf{X}^{T}$, let 
$$
(X_{1:T},\tilde{X}_{1:T}) \gets \textsc{CoupledCBPF}(x_{1:T}^*,\tilde{x}_{1:T}^*, N),
$$
and denote by 
$b^* = \sum_{t=1}^T \I(x_t^* \neq \tilde{x}_t^*)$ and
$B = \sum_{t=1}^T \I(X_t \neq \tilde{X}_t)$ the number of unequal input and output states, respectively. 

Under \textup{\ref{a:strong-mixing}} and \textup{\ref{a:joint-maximal-coupling}}, 
there exist finite constants $N_{\rm min}$ and $c_\lambda$ depending only on the constants in \textup{\ref{a:strong-mixing}} such that for all $N\ge N_{\rm min}$:
\begin{equation}
\E[B] \le \lambda_N b^*, \qquad \text{where}\qquad \lambda_N \le c_\lambda \frac{\log^2 N}{\sqrt{N}}.
\label{eq:lambda-upper-bound}
\end{equation}
\end{theorem}
The proof of Theorem \ref{thm:main-contraction} is given in Appendix \ref{app:proof-main-contraction}.

We conclude this section by showing how Theorem \ref{thm:main-contraction} leads to an upper bound for the meeting time of the iterated coupled CBPF, proving Theorem \ref{thm:cbpf-mixing}.
\begin{corollary}
  \label{cor:coupling-time}
Let $(\vec{S}_0, \tilde{\vec{S}}_0) \in \mathsf{X}^{T}\times\mathsf{X}^{T}$ be arbitrary, and for $k\ge 1$:
  $$
  (\vec{S}_k,\tilde{\vec{S}}_k) \gets \textsc{CoupledCBPF}(\vec{S}_{k-1},\tilde{\vec{S}}_{k-1}, N).
  $$
Under \textup{\ref{a:strong-mixing}} and \textup{\ref{a:joint-maximal-coupling}}, the distribution of the meeting time 
  $$
  \tau = \inf\{ k \given \vec{S}_k = \tilde{\vec{S}}_k \},
  $$
is upper bounded as follows:
$$
\P(\tau > k) \le T \lambda_N^k,
$$
where $\lambda_N$ is given in Theorem \ref{thm:main-contraction}, and satisfies the upper bound \eqref{eq:lambda-upper-bound}.
\end{corollary}
\begin{proof}
  Denote by $B_k = \sum_{t=1}^T \I\big([\vec{S}_k]_t \neq [\tilde{\vec{S}}_k]_t\big)$ the number of uncoupled states at iteration $k$ and by $\mathcal{F}_k = \sigma(\vec{S}_i, \tilde{\vec{S}}_i\given i\le k)$ the history. By Theorem \ref{thm:main-contraction}:
  \begin{align*}
  \E[B_k] 
  &= \E\big[\E[B_k\mid \mathcal{F}_{k-1}]\big] 
  \le \E[\lambda_N B_{k-1}] \le \cdots \le T \lambda_N^k,
  \end{align*}
  because $B_0\le T$. Note also that if $B_{k-1}=0$, then $B_k=0$, so $\{B_k=0\} = \{\tau\le k\}$, so the claim follows from Markov's inequality:  $\P(B_k\ge 1) \le \E[B_k]$.
\end{proof}

\begin{proof}[Proof of Theorem \ref{thm:cbpf-mixing}]
  Let $\vec{S}_0 = x_{1:T}^*\in \X^T$ be arbitrary and let $\tilde{\vec{S}}_0\sim \pi$. Then,
  $$
     \| \mathrm{Law}_N(\vec{S}_k) - \pi \|_\mathrm{TV} 
     = \| \mathrm{Law}_N(\vec{S}_k) - \mathrm{Law}_N(\tilde{\vec{S}}_k) \|_\mathrm{TV}
     \le \P(\tau > k)
     \le T \lambda_N^k,
  $$
where the last inequality follows from Corollary \ref{cor:coupling-time}. The claim follows with $r_N = \lambda_N^{-1}$.
\end{proof}

Note that the upper bound $\P(\tau > k)$ of the total variation distance would be tight if we had analysed a maximal coupling of the marginals $\mathrm{Law}_N(\vec{S}_k)$ and $\mathrm{Law}_N(\tilde{\vec{S}}_k)$. The coupling implemented by \textsc{CoupledCBPF} appears to be sufficiently good to achieve the right mixing rate in $T$, as confirmed by the example in Proposition \ref{prop:logTexample}.

\begin{remark}
  The proof of Theorem \ref{thm:cbpf-mixing} implies also that the laws of individual coordinates $[\vec{S}_k]_t$ converge \emph{on average} uniformly in $T$. Namely, let $\proj_t(\pi)$ stand for law of the time marginal $t$ of $\pi$, then
  $$
  \frac{1}{T} \sum_{t=1}^T \| \mathrm{Law}_N([\vec{S}_k]_t) - \proj_t(\pi) \|_\mathrm{TV}
  \le \frac{1}{T}\E[B_k] \le \lambda_N^k,
  $$
  because $\| \mathrm{Law}_N([\vec{S}_k]_t) - \proj_t(\pi) \|_\mathrm{TV} \le \P\big([\vec{S}_k]_t \neq [\tilde{\vec{S}}_k]_t\big)$.
\end{remark}

\section{Alternative couplings and unbiased estimators}
\label{sec:unbiased}

Implementable couplings of Markov chains have been used recently in unbiased estimation, following the idea suggested by \cite{glynn-rhee}; see \citep{jacob-oleary-atchade} and references therein. Couplings of CPFs were suggested for this purpose in \citep{jacob-lindsten-schon} and a special case of the coupled CBPF (Algorithm \ref{alg:mc_cbpf}) was suggested in \citep{lee-singh-vihola}.

Algorithm \ref{alg:unbiased} describes one way in which Algorithms \ref{alg:cbpf} and \ref{alg:mc_cbpf} can be used to produce `$L$ lagged $k$ offset' unbiased estimators. In line \ref{line:init1}, we suggest to initialise with the particle filter (Algorithm \ref{alg:cbpf} omitting the reference particles at index $i=0$), but any other initialisation is valid as long as $\tilde{\mathbf{S}}_0$ has the same distribution as $\mathbf{S}_{-L}$.
\begin{algorithm}
  \caption{\textsc{Unbiased}($h, \mathbf{s}_0, N, k, L$)}
  \label{alg:unbiased} 
\begin{algorithmic}[1]
  \State  Set $\tilde{\mathbf{S}}_0 = \mathbf{S}_{-L} = \mathbf{s}_0 \gets \textsc{ParticleFilter}(N)$ \Comment{Initialise} \label{line:init1}
  \State  Set $\mathbf{S}_j \gets
  \textsc{CBPF}(\mathbf{S}_{j-1},N)$ for $j=1-L,\ldots,0$ \Comment{Advance $\mathbf{S}_{-L}$ by $L$ steps} \label{line:init2} 
  \For{$n=1,2,\ldots$}
  \State  $(\mathbf{S}_{n},\tilde{\mathbf{S}}_{n}) \gets
  \textsc{CoupledCBPF}(\mathbf{S}_{n-1},\tilde{\mathbf{S}}_{n-1},N)$
  \Comment{Coupled transitions}
  \If{$\mathbf{S}_n = \tilde{\mathbf{S}}_n$ and $n\ge k$}
  \State \textbf{output} $Z_k = h(\mathbf{S}_k) + \sum_{j=1}^{\lfloor(n-k)/L\rfloor} [h(\mathbf{S}_{k+Lj}) -
  h(\tilde{\mathbf{S}}_{k+Lj})]$ \label{line:stopping}
  \EndIf
  \EndFor
\end{algorithmic}
\end{algorithm}

Hereafter, let $\tau_k$ stand for the number of iterations of \textsc{CoupledCBPF} in Algorithm \ref{alg:unbiased}, that is, the value of $n$ at the time it terminates. If $\tau_k$ is almost surely finite and suitably well-behaved, then for a large class of functions $h:\X^T\to\R$ the output is unbiased: $\E[Z_k] = \int h(\vec{s}) \pi(\vec{s}) \ud \vec{s}$ \citep[see, e.g,][]{jacob-oleary-atchade,lee-singh-vihola,douc-jacob-lee-vats}. It is often useful to consider a number of different offset parameters $k$ simultaneously, and use an average $Z_{k:\ell}$ of $Z_k,\ldots,Z_{\ell}$ in estimation \citep{jacob-oleary-atchade,douc-jacob-lee-vats}. A practical rule of thumb for the tuning parameters $L$, $k$, and $\ell$ is to set $L$ to a high quantile of the meeting times, $k=L$ and $\ell = 5k$ \citep{douc-jacob-lee-vats}.

The efficiency of the unbiased estimation in Algorithm \ref{alg:unbiased} depends on the behaviour of the running time $\tau_k$, as well as the computational cost of each iteration of the \textsc{CoupledCBPF}. Different instances of Algorithm \ref{alg:mc_cbpf} arise when \textsc{FwdCouple} is implemented in a different manner. We detail four choices, some of which are new, in the subsections below.

\subsection{Joint maximal coupling}

In the theoretical analysis above, we have assumed that \textsc{FwdCouple} is equivalent to $\textsc{MaxCouple}\big((\fsp_t^N)^{\otimes N},(\tilde{\fsp}_t^N)^{\otimes N}\big)$, where $\fsp_t^N$ and $\tilde{\fsp}_t^N$ are the predictive distributions of the two CPFs, given in \eqref{eq:filter-state-predictive}. Hereafter, we call this strategy `joint maximal coupling' (JMC). A similar coupling strategy was explored in our earlier theoretical work \citep{karjalainen-lee-singh-vihola}, in the context of particle filters.

In order to simulate from JMC, we need to employ Algorithm \ref{alg:maxcouple-generic} in Appendix \ref{app:maxcouple}, which simulates from $\textsc{MaxCouple}(p,q)$, a maximal coupling of any $p$ and $q$ probability distributions (densities). Algorithm \ref{alg:maxcouple-generic} is a rejection sampler, which requires expected $O(1)$ simulations from and/or evaluations of $p$ and $q$. In the context of JMC, $p$ and $q$ are $N$-fold products of predictive densities $\fsp_t^N$ and $\tilde{\fsp}_t^N$ and these mixtures themselves require $O(N)$ operations to evaluate. The overall cost of JMC is therefore $O(N^2)$.

Corollary \ref{cor:coupling-time} guarantees that whenever $N$ is sufficiently large, $\tau_k$ is $O(k\vee \log T)$ and has exponential tails. Therefore, taking suitable $k = O(\log T)$ implies that $Z_k$ has nearly ideal variance: $\mathrm{var}(Z_k) \approx \mathrm{var}_\pi(h)$ (the variance of $h(X)$ where $X\sim \pi$), at least for bounded $h$ \citep[see][Lemma 31]{lee-singh-vihola}.

\subsection{Independent maximal coupling}

We can also consider maximal couplings of individual particles, that is, let \textsc{FwdCouple} stand for $\textsc{MaxCouple}(\fsp_t^N,\tilde{\fsp}_t^N)^{\otimes N}$. This means that each particle pair is simulated \emph{independently} from a maximal coupling of their marginal laws:
$$
(X_t^i,\tilde{X}_t^i) \sim \textsc{MaxCouple}(\fsp_t^N,\tilde{\fsp}_t^N).
$$
This strategy, which we call `independent maximal coupling' (IMC), has been suggested earlier with particle filters applied in the multilevel Monte Carlo context \citep{jasra-yu}, but IMC has not been explored with coupled CPFs before. 
Like JMC, simulating from IMC requires using Algorithm \ref{alg:maxcouple-generic} in Appendix \ref{app:maxcouple}, and has $O(N^2)$ overall cost.

Our empirical findings in Section \ref{sec:experiments} suggest that IMC behaves similarly to JMC, and in fact, can have improved performance.
We suspect that IMC admits similar theoretical properties as JMC, that is, $O(k \vee \log T)$ coupling time for sufficiently large $N$. 
However, despite several attempts, we have not managed to show this theoretically, so we can only conjecture its coupling time to be $O(k \vee \log T)$.

\subsection{Independent index coupling}
\label{sec:ind-ix-coupling}

The first implementation of CPF couplings was suggested in \citep{jacob-lindsten-schon}, based on coupled resampling, where each pair of particles is drawn with two-phase algorithm:
\begin{enumerate}[(i)~]
  \item $(A^i,\tilde{A}^i) \sim \textsc{MaxCouple}(\mathrm{Categorical}(W^{0:N}), \mathrm{Categorical}(\tilde{W}^{0:N}))$
  \item If $X_{t-1}^{A^i} = \tilde{X}_{t-1}^{\tilde{A}^i}$, then $X_t^i = \tilde{X}_t^i \sim M_t(X_{t-1}^{A^i},\uarg)$; \\ 
  otherwise $X_t^i \sim M_t(X_{t-1}^{A^i},\uarg)$ and $\tilde{X}_t^i \sim M_t(\tilde{X}_{t-1}^{\tilde{A}^i},\uarg)$ independently.
\end{enumerate}
The same procedure 
was introduced earlier in \citep{chopin-singh} for theoretical reasons: to prove the uniform ergodicity of the CPF. 
This algorithm was used within \textsc{CoupledCBPF} (Algorithm \ref{alg:mc_cbpf}) in \citep{lee-singh-vihola}.
In the present paper, we refer to this method as `independent index coupling' (IIC).

The indices $A^{1:N}$ and $\tilde{A}^{1:N}$ can be drawn in $O(N)$ time \citep[cf.][]{jacob-lindsten-schon,lee-singh-vihola}.
When IIC is used, it was shown in \citep[Theorem 6]{lee-singh-vihola} that there exists $N_0$ only depending on \ref{a:strong-mixing} such that for all sufficiently large $N$, $\tau_k = O(k \vee T)$, and has exponential tails. This means that when IIC is used, taking $k = O(T)$ guarantees that $\tau_k = O(T)$ and for bounded $h$, unbiased estimator $Z_k$ from Algorithm \ref{alg:unbiased} has $\mathrm{var}(Z_k) \approx \mathrm{var}_\pi(h)$ \citep{lee-singh-vihola}.

\subsection{Joint index coupling}

Analogous to IMC/JMC, we can also consider a `joint index coupling' algorithm (JIC), which is the maximal coupling of the resampling indices:
$$
   (A^{1:N},\tilde{A}^{1:N}) \sim \textsc{MaxCouple}\big(\mathrm{Categorical}(W^{0:N})^{\otimes N}, 
   \mathrm{Categorical}(\tilde{W}^{0:N})^{\otimes N}\big).
$$
The JIC has $O(N)$ complexity like IIC, because sampling from the product is the usual resampling operation, which is $O(N)$. Products of categorical distributions are proportional to the multinomial distribution, and the ratios $q(A^{1:N})/p(A^{1:N})$ are of the form $\prod_{i=1}^N (\tilde{v}^{A^i}/v^{A^i})$, where $\tilde{v}^i,v^i$ are normalised  $\tilde{W}^i, W^i$, respectively.

We believe that JIC has similar theoretical properties as IIC was shown to satisfy in \citep{lee-singh-vihola}. Indeed, it would be relatively straightforward to modify the proof \citep[Theorem 6]{lee-singh-vihola} to accommodate JIC. However, we do not pursue this further, because our experiments (see Section \ref{sec:experiments}) suggest that JIC is generally less efficient than IIC. 

\subsection{Computational considerations and hybrid strategies}

It is not difficult to see that whenever the CPF states coincide, that is, $X_{t-1}^{0:N} = \tilde{X}_{t-1}^{0:N}$, then all IMC/JMC/IIC/JIC produce $X_{t}^{1:N} = \tilde{X}_{t}^{1:N} \sim \zeta_t^N$. The CPF states can be checked in an implementation, and whenever they coincide, one can simulate once from $\zeta_t^N$ and copy, thereby avoiding a costly coupled sampling step. This can provide a substantial speedup in particular for $O(N^2)$ complexity IMC and JMC.

Nevertheless, the $O(N^2)$ complexity of JMC and IMC makes them less efficient for large $N$. The algorithms can still be useful in practice because a moderate $N$ can be sufficient and calculations are easily parallelised. It might also be possible to develop algorithms that implement \emph{nearly} maximal couplings, which enjoy similar properties, but which have substantially reduced computational cost.

The \textsc{CoupledCBPF} remains valid also if the coupling strategy varies from one time instant to another. It can also depend on the previous filter states and reference trajectories. We believe that it is possible to accommodate our theoretical results for certain strategies that force JMC to be used, say, every $\ell$th step and whenever the non-reference particles are all equal. We have not investigated such `hybrid' strategies further.

\subsection{Potentials with pairwise dependencies}

The coupled CBPF (Algorithm \ref{alg:mc_cbpf}) is applicable only with a Feynman--Kac representation of the form \eqref{eq:feynman-kac-pi}. In particular, the potential $G_t(x_t)$ can only depend on the current state, which means in the HMM setting that we use a `bootstrap filter' \citep{gordon-salmond-smith} like strategy: $M_t$ corresponds to the HMM latent state dynamics.

The earlier coupling method \citep{lee-singh-vihola} allowed for dependence on the previous state, too, having $G_t(x_{t-1},x_t)$, which allows for using more flexible (and potentially more efficient) proposal distributions. We describe a generalisation of Algorithm \ref{alg:mc_cbpf} to potentials with such pairwise dependencies in Appendix \ref{app:pairpotcoup}. The method has the same $O(N^2)$ complexity as IMC/JMC.

\section{Experiments}
\label{sec:experiments}

We investigated empirically the behaviour of iterated \textsc{CoupledCBPF}, or more specifically, Algorithms \ref{alg:mc_cbpf} and \ref{alg:unbiased}, with four models: a simple model on a torus with `barriers', a linear-Gaussian model, a stochastic volatility model and a model for inferring neural spike trains from calcium fluorescence imaging data. Two of the models, the first and the last, satisfy the strong mixing condition \ref{a:strong-mixing}. We tested different choices for \textsc{FwdCoupling} (IMC, JMC, IIC and JIC) discussed in Section \ref{sec:unbiased}.
The source code is available at \url{https://github.com/mvihola/CBPFCouplingPaperCodes}.

\subsection{Barriers on a torus}
\label{sec:strongmixing}

Our first model is a simple model on $\X=[0,1]$ which satisfies the strong mixing condition \ref{a:strong-mixing}. The dynamic model is a mixture of uniform distribution $U(0,1)$ with probability $a\in(0,1)$, and a random walk with uniform $U(-w/2,w/2)$ increments with probability $1-a$. The transition density is:
$$
   M_t(x,y) = \begin{cases}
    a + \frac{1-a}{w}, & \text{if }d(x,y)\le \frac{w}{2}, \\
    a, & \text{otherwise,}
   \end{cases}
$$
where $d(x,y)$ is the metric on the torus: $d(x,y) = |x-y|$ if $|x-y|<0.5$ and $1-|x-y|$ if $|x-y|>0.5$. We used the increment distribution width $w=0.2$ in the experiments and varied $a\in\{0.1,0.3,0.5\}$, ranging from nearly a random walk to uniform transitions. The potentials of the model were all set to:
$$
  G_t(x) = \begin{cases}
    b, & x\in [0,\frac{1}{4}] \cup (\frac{1}{2},\frac{3}{4}] \\
    1-b, &\text{otherwise},
  \end{cases}
$$
where the parameter $b\in \{0.1,0.3,0.5\}$ encodes how big the `barriers' the $G_t$ creates are: with $b=0.5$ the potentials $G_t$ are constant and small $b$ introduces disjoint intervals, which can be challenging to bridge when $a$ is small.

We used the model with time horizons $T=2^9=512,\ldots, 2^{14}=16384$ and ran Algorithm \ref{alg:unbiased} 100 times (with $L=1$ and $k=0$) to investigate coupling times. We varied the number of particles $N+1 = 2^1, \ldots, 2^7$, and capped the computation time of each configuration to a pre-defined limit of 4 hours.

Figures \ref{fig:torusIterations} and \ref{fig:torusIterations2} 
\begin{figure}
  \includegraphics[scale=0.9]{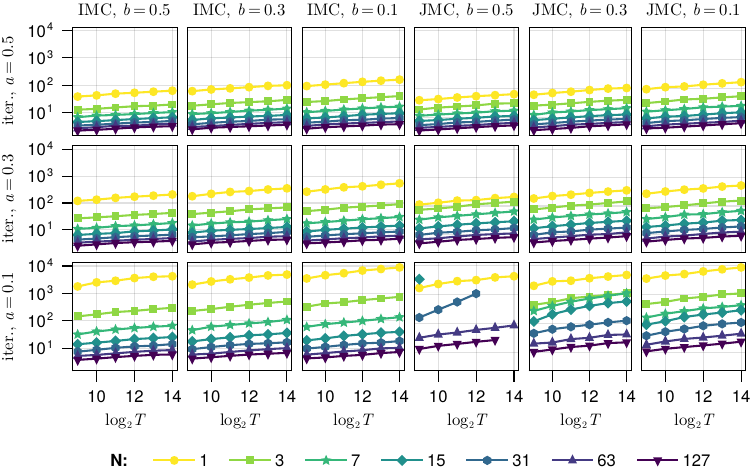}
\caption{Average coupling times in Algorithm \ref{alg:unbiased} (using IMC and JMC) for the barriers model described in Section \ref{sec:strongmixing}. The experiments that failed to complete within the time limit of 4 hours (with $a$ = 0.1, $b$ = 0.5) are omitted from the graphs.}
  \label{fig:torusIterations}
\end{figure}
\begin{figure}
  \includegraphics[scale=0.9]{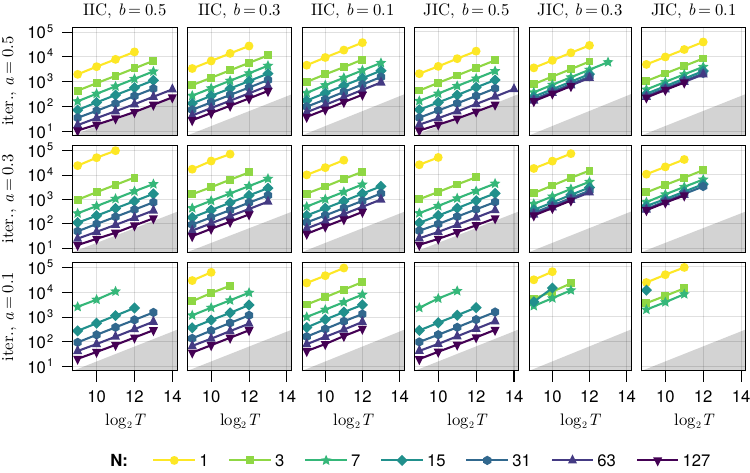}
\caption{Average coupling times in Algorithm \ref{alg:unbiased} (using IIC and JIC) for the barriers model described in Section \ref{sec:strongmixing}. As above, the experiments that failed to complete within the time limit of 4 hours are omitted from the graphs. The grey areas illustrate a linear growth rate.}
  \label{fig:torusIterations2}
\end{figure}
show the number of iterations until coupling for the state couplings (IMC and JMC) and index couplings (IIC and JIC), respectively.
With $a=0.3$ and $a=0.5$, the number of iterations grows slowly with respect to the time horizon $T$ for both IMC and JMC, as anticipated by our theory for the latter. In contrast, with IIC and JIC, the number of iterations appears to increase linearly in $T$, as expected \citep[cf.][]{lee-singh-vihola}. Some experiments with long time horizon $T$ failed to complete within the time limit.

With $a=0.1$ and $b=0.5$, JMC failed to complete with $N=3$ and $N=7$. With $N=31$, JMC completed for the smaller $T$ but the number of iterations appears to increase at a faster (linear) rate. In this scenario, the sufficiently large $N$ for `stable' behaviour appears to be $N \ge 63$. With $N=1$, JMC coincides with IMC, and also completed succesfully. With $a=0.1$ and $b=0.5$, the IIC and JIC behaved similarly. For $a=0.1$ and $b\le 0.3$, JIC struggled to complete with larger $N$. 

Because the index couplings (IIC and JIC) have $O(N)$ computational complexity and state couplings (IMC and JMC) $O(N^2)$, a fair comparison needs to account for this. We multiplied the number of iterations by $N$ or $N^2$, respectively, to get average `cost' factors.
Figure \ref{fig:torusCost}
\begin{figure}
  \includegraphics[scale=0.9]{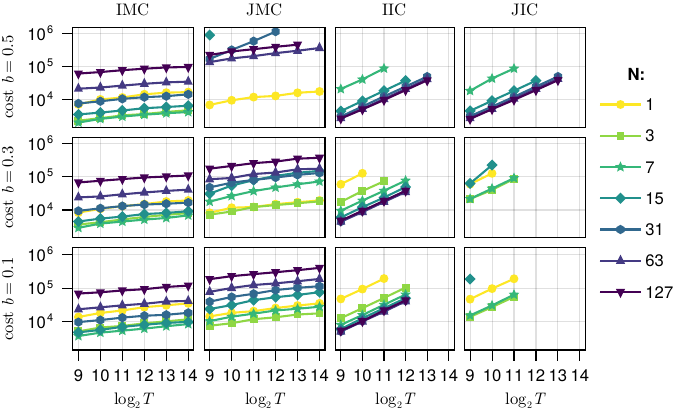}
\caption{Average cost factors for Algorithm \ref{alg:unbiased} in the barriers model described in Section \ref{sec:strongmixing} with $a=0.1$. The experiments that failed to complete within the time limit of 4 hours are omitted from the graphs.}
  \label{fig:torusCost}
\end{figure}
shows the costs for scenarios with $a=0.1$. For shorter time horizons, the index couplings, and in particular the IIC, can be competitive, but for larger $T$, state couplings with suitable $N$ are clearly favourable. When the $O(N^2)$ complexity is accounted for, it becomes evident that $N$ should be chosen carefully for IMC and JMC to achieve the best efficiency.

\subsection{Linear-Gaussian model}
\label{sec:lgmodel}

Our second model is a simple linear-Gaussian model where $(M_t)$ corresponds to stationary AR(1) process $X_t = \rho X_{t-1} + \sigma_X W_t$, with noisy Gaussian observations that are all zero, leading to $G_t \propto \exp(-0.5x^2/\sigma_Y^2)$.
The model has parameters $\theta=(\rho,\sigma_X,\sigma_Y)$ where $\rho\in(-1,1)$ and $\sigma_X,\sigma_Y>0$.

We consider the model with parameters $\theta_1 = (0.9, 1.0, 1.0)$, $\theta_2 = (.99,0.105,1.0)$ and $\theta_3 = (0.99, 0.105, 10)$.
The parameters $\theta_1$ and $\theta_3$ may be viewed as corresponding to time-discretisations of the same continuous-time Ornstein--Uhlenbeck latent process with Gaussian path integral weights, with $\theta_3$ corresponding to ten-fold finer time stepping.
The model with parameter $\theta_3$ is expected to be more challenging, because the (basic) particle filter struggles with finer time-discretisations \citep[cf.][]{chopin-singh-soto-vihola} and $\theta_2$ is an intermediate between $\theta_1$ and $\theta_3$.

We considered the same time horizons $T$ and numbers of particles $N$ as for the barriers model, and investigated averages of 1000 repeated runs of Algorithm \ref{alg:unbiased} capping computation time to 8 hours.

Figure \ref{fig:lgCost}
\begin{figure}
  \includegraphics[scale=0.9]{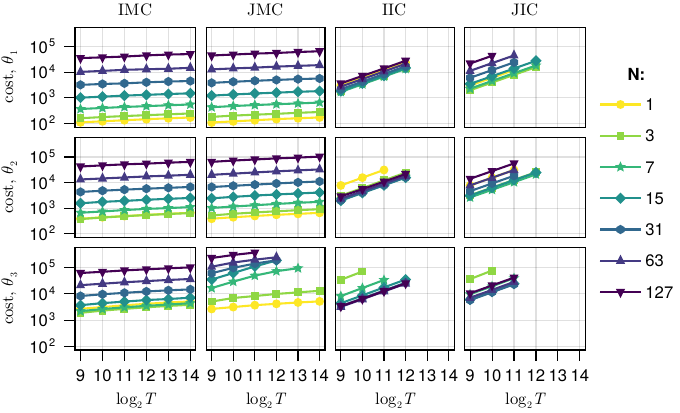}
  \caption{Average cost factors for Algorithm \ref{alg:unbiased} in the linear-Gaussian model described in Section \ref{sec:lgmodel} with parameter configurations $\theta_1$, $\theta_2$ and $\theta_3$. The experiments that failed to complete within the time limit of 8 hours are omitted from the graphs.}
  \label{fig:lgCost}
\end{figure}
shows the average cost factors. The state couplings (IMC and JMC) perform similarly and appear to outperform the index couplings (IIC and JIC) with model parameters $\theta_1$ and $\theta_2$. The number of iterations required by IIC/JIC appears to increase linearly with respect to $T$, in contrast with a very slow increase for IMC/JMC, which appears to be very similar to the results with the barriers model.
The companion of Figure \ref{fig:lgCost} that records iterations instead of cost factors is given in Figure \ref{fig:lgIterations} in Appendix \ref{app:extraexperiments}. For the most challenging model with parameter $\theta_3$, JMC with $N=7$ and $N=15$ led to coupling times which appear to increase faster with respect to $T$. IMC and JMC appear to be generally competitive against IIC/JIC with smaller $T$, and outperform the index couplings when $T$ is large. 

We investigated also in more detail how the couplings develop when Algorithm \ref{alg:mc_cbpf} is iterated. Figure \ref{fig:couplingPlots} shows one realisation of the iterated coupling algorithms for the model with parameter $\theta_3$: black pixel at location $(i,t)$ indicates that the states at iteration $i$ and time index $t$ differ: $X_t(i) \neq \tilde{X}_t(i)$. The IMC with $N=63$ and JMC with $N=255$ appear to be `stable' and couple quickly. This is how we expect JMC to behave if Theorem \ref{thm:main-contraction} were in force and $N$ was large enough.
In contrast, IMC with $N=1$ and JMC with $N=7$ appear to be `critical' coupling systems: they eventually couple, but only after hundreds of iterations. In particular, JMC with $N=7$ leads to branching random walk type pattern, which drifts to the right, while IMC with $N=1$ appears similar to IMC with $N=63$ or JMC with $N=255$, but coupling takes much longer. JMC with $N=7$ has `spontaneous' couplings everywhere, but they tend to break from one iteration to another, except for the  `frontier' of couplings that builds progressively from the start. This progressive behaviour is similar to IIC with $N=15$, where couplings can typically only occur close to the frontier; see the theoretical analysis in \citep{lee-singh-vihola}. 
\begin{figure}
  \includegraphics[width=11.43cm]{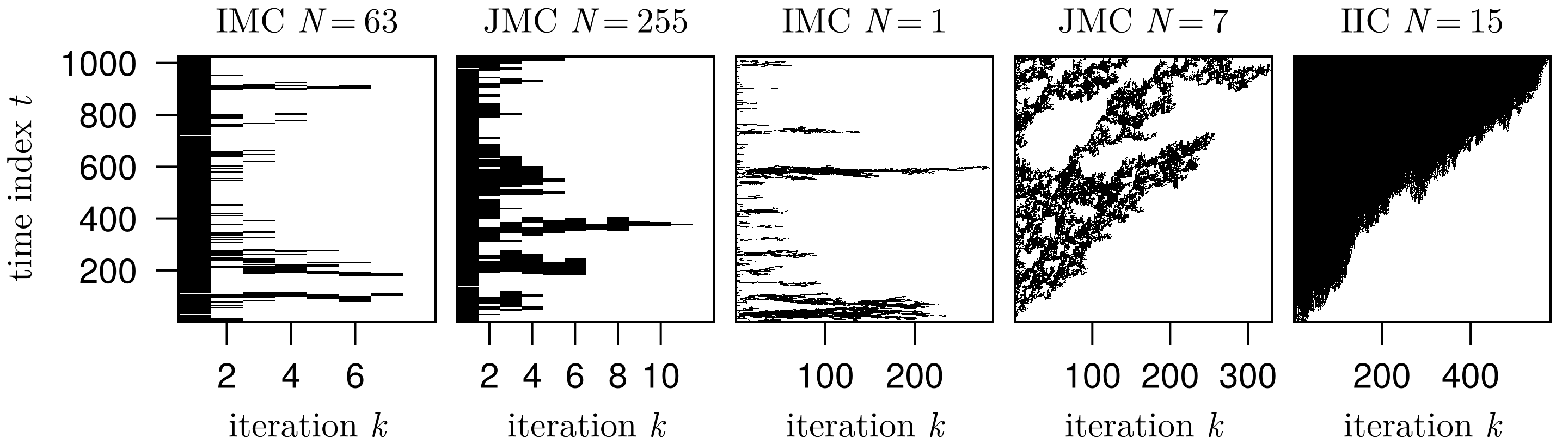}
\caption{Illustration of the couplings as Algorithm \ref{alg:mc_cbpf} is iterated for the linear-Gaussian model (parameter configuration $\theta_3$) described in  Section \ref{sec:lgmodel}. Uncoupled states, i.e. $X_t(i) \neq \tilde{X}_t(i)$, are shown as black pixels.}
  \label{fig:couplingPlots}
\end{figure}

\subsection{Stochastic gradient maximum likelihood}
\label{sec:mle}

We apply the coupling methods to stochastic gradient maximum likelihood estimation. We find parameters of a stochastic volatility (SV) model with leverage, with real financial returns data (log returns corresponding to the MSCI Switzerland Index data from R package AER \citep{kleiber-zeileis}) with a total of $T=4696$ data points.

The SV model, from \citep{omori-chib-shephard-nakajima}, has latent log-volatility $(X_t)$ for observed log-returns $Y_t$:
$$
\begin{aligned}
Y_t &= \epsilon_t \exp(X_t/2)  \\
X_{t+1} &= \mu + \phi (X_t - \mu)  + \eta_t,
\end{aligned}
\quad
\text{where}
\quad
\begin{bmatrix}
  \epsilon_t \\
  \eta_t 
  \end{bmatrix} 
  \sim N\left(\begin{bmatrix}0 \\ 0\end{bmatrix},   
  \begin{bmatrix}1 & \rho\sigma \\
  \rho\sigma & \sigma^2 
\end{bmatrix}\right)
$$
The parameters of the model $\vec{\theta} = (\mu,\phi,\rho,\sigma)$ are the mean volatility $\mu\in\R$, the autoregressive (AR) model coefficient $\phi \in (-1,1)$ and variance $\sigma>0$, and the noise correlation parameter $\rho \in (-1,1)$. Because of the noise correlation, the model is not a HMM, but it can be written in the following Feynman--Kac model form \citep{omori-chib-shephard-nakajima}:
\begin{align*}
  G_t^\theta(x_t) &= p(y_t\mid x_t) = f_N(y_t; 0, e^{x_t}) \\
  M_{t+1}^\theta(x_{t}, x_{t+1}) &= p(x_{t+1}\mid x_{t}, y_t) = f_N(x_{t+1}; \mu_{t+1}(x_{t}), (1-\rho^2)\sigma^2)
\end{align*}
where $f_N(x; \mu,\sigma^2)$ is the density of $N(\mu,\sigma^2)$ at $x$  and $\mu_{t+1}(x_{t}) = \mu + \phi(x_{t}-\mu) + \rho\sigma e^{-x_{t}/2}y_t$ corresponds to the conditional expectation of $x_{t+1}$ given $x_{t}$ and $y_t$.
Furthermore, $M_1(x_1) = f_N(x_1; \mu, \sigma_s^2)$ where $\sigma_s^2 = \sigma^2/(1-\rho^2)$ is the stationary variance of the AR process.

The joint density $p(x_{1:T},y_{1:T})$ under this model is
$$
   \gamma^\theta(x_{1:T}) = M_1^\theta(x_1) G_1^\theta(x_1) \prod_{t=2}^T M_{t}^\theta(x_{t-1},x_t) G_t^\theta(x_t),
$$
and the likelihood of observations $y_{1:T}$ is $L(\theta) = \int \gamma^\theta(x_{1:T}) \ud x_{1:T}$.
We employ a stochastic gradient ascent algorithm in order to iteratively find the maximiser of $L$; see \citep{kantas-etal} for discussion on different (approximate) estimation methods for HMMs (or Feynman--Kac) models.

The conditional density $p(x_{1:T}\mid y_{1:T}) = \pi^\theta(x_{1:T}) = \gamma^\theta(x_{1:T})/L(\theta)$. The score can be written (by Fisher's identity) as an expectation with respect to $\pi^\theta$:
\begin{equation*}
\nabla_{\vec{\theta}} \log L(\theta)
= \int h^\theta(x_{1:T}) \pi^\theta(x_{1:T}) \ud x_{1:T}, \quad \text{with}\quad
h^\theta(x_{1:T}) = \nabla_{\vec{\theta}} \log \gamma^\theta(x_{1:T}).
\end{equation*}
We calculate unbiased estimates of $\nabla_{\vec{\theta}} \log L(\theta)$ using Algorithm \ref{alg:unbiased}, and feed the gradients to the Adam optimiser \citep{kingma-ba} to iteratively maximise the likelihood.

We optimise the log and logit transformed parameters, taking values in $\R^4$, with the initial choices $\mu = \mathrm{var}(y_{1:T})$, $\sigma=1$ and $\rho=\phi=0$. We use $N=16$ particles in all tests, and perform a preliminary run comprising 100 repeated runs of the coupling algorithms to determine a 90\% empirical quantile $q$, which is used to set the parameters of the lagged averaged unbiased estimator $Z^\theta_{k:\ell}$. Specifically, lag $L=q$, $k=q$ and $\ell=5q$, following the recommendations in \citep{douc-jacob-lee-vats} (see Section \ref{sec:unbiased}). The parameters of Adam are set to defaults as advised in \citep{kingma-ba}, except for learning rate which is set to $\alpha=0.01$.

Figure \ref{fig:mle-couplingtimes}  (left column) shows the stochastic gradient algorithm's parameter estimates $\hat{\theta}_t$ with respect to elapsed (wall clock) time. With IMC, convergence happens before 10 minutes have elapsed. JMC also gets near optimum values quickly, but then slows down (and fluctuates). With index couplings (IIC and JIC), convergence is substantially slower, and we did not observe it within the recorded time. However, all methods appear to follow a similar path when plotted with respect to iterations (results not shown), which suggests that all methods can ultimately recover the (approximate) MLE.
\begin{figure}
  \includegraphics[width=\linewidth]{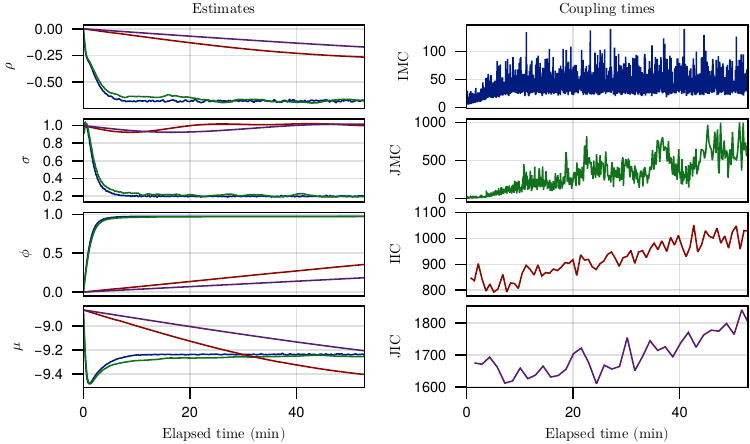}
  \caption{ Iterates of the maximum likelihood estimates (left) for the stochastic volatility model with leverage (Section \ref{sec:mle}) using the MSCI Switzerland Index data, and the observed coupling times (right).}
  \label{fig:mle-couplingtimes}
\end{figure}

The right column of Figure \ref{fig:mle-couplingtimes} shows the observed coupling times, that is, iterations until we observe full coupling to happen. The distribution of coupling times changes with respect to iterations, because of the changing model. With all forward coupling methods, the coupling times are smaller with the initial parameter values and slow down closer to the MLE. While JMC is competitive with IMC initially, the JMC gets much worse, and in fact sometimes hit the pre-defined maximum number of iterations 1000. The changing coupling times suggests that it could be useful to adjust the parameters $L$, $k$ and $\ell$ iteratively.

Figure \ref{fig:msci-fit} shows the MSCI data, the estimated logarithmic volatility corresponding to the model with the MLE $\hat{\theta}\approx(-9.24,0.97,-0.67,0.20)$ (averaged from the last 10\% of IMC estimates), and a standard normal quantile plot of $\hat{\epsilon}_t = y_t e^{-X_t/2}$, where $X_{1:T} \sim \pi^{\hat{\theta}}$, indicating a good fit. The $\hat{\epsilon}_t$ did not have substantial autocorrelations either (results not shown).
\begin{figure}
  \includegraphics[width=\linewidth]{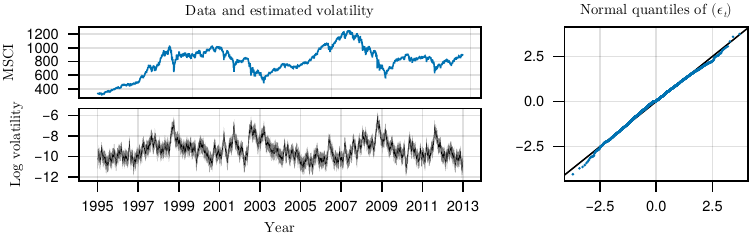}
  \caption{The MSCI data, the estimated log-volatilities with MLEs (median and 95\% confidence band) and normal quantile plot for $\hat{\epsilon}_t$ (see Section \ref{sec:mle}).}
  \label{fig:msci-fit}
\end{figure}

\subsection{Calcium fluorescence imaging}
\label{sec:fluorescence}

Our next example is a model for inferring neuronal spike trains from two-photon fluorescence observations of calcium ($\mathrm{Ca}^{2+}$) concentration \citep{vogelstein-etal}. An inference objective is the spike times, which
is assumed to follow a Poisson process with rate parameter $p$, that is, $N_t \sim \mathrm{Poisson}(p\Delta)$, where $\Delta = 0.01$ (seconds) is the time step. Another inference objective is the calcium concentration $C_t$, which is assumed to follow the model
$$
   C_t = \bigg( 1 - \frac{\Delta}{\tau}\bigg) C_{t-1}  + \frac{\Delta}{\tau} c_{0} + a N_{t-1} + \sigma_C \sqrt{\Delta} \epsilon_{C,t}, \quad \epsilon_{C,t} \sim N(0,1)
$$
where $\tau>0$, $c_0$, $a>0$ and $\sigma_C>0$ are parameters. The fluorescence observations are assumed to be noisy observations of calcium concentrations:
$$
  F_t = \alpha C_t + \beta + \sigma_F \epsilon_{F,t} \quad \epsilon_{F,t}\sim N(0,1),
$$
where $\alpha,\beta\in\mathbb{R}$ and $\sigma_F,\xi>0$ are parameters. The latent state is  $X_t = (C_t, N_t)$,  the calcium concentration and spike count. The units of $C_t$ and $F_t$ are unimportant for the task, so we assume $\alpha=1$ and $\beta=0$, rescale the data to $[0,1]$, and assume the model of $F_t$ and $C_t$ to be truncated on this interval. We also truncate $N_t$ to be at most ten.

We consider real two-photon imaging data \citep{cai3-dataset,theis-etal}. We fix the parameter $p\approx0.38$ based on the reported true spike rate of the data, and the observation noise $\sigma_F \approx 0.028$ based on residuals of the signal with respect to its trend (computed with a moving average) from a period without spikes. We estimate the rest of the model parameters using $10,000$ iterations of a stochastic gradient method similar to Section \ref{sec:mle}. We estimate the gradients at $\theta^{(k)}$ by $h^{\theta^{(k)}}(X_{1:T}^{(k+1)})$, where the updates $X_{1:T}^{(k)}\to X_{1:T}^{(k+1)}$ come from one CBPF iteration with parameters $\theta^{(k)}$ and $N=64$ particles; this is an instance of stochastic approximation with Markovian noise \citep[cf.][]{benveniste-metivier-priouret}. The Feynman-Kac model coincides with the HMM above, except for that we used fixed Poisson distribution with rate $0.03$ to propose jumps for $N_t$. It is straightforward to verify that this model satisfies \ref{a:strong-mixing}.
The parameters were estimated to be $\sigma_C\approx 0.05$, $\tau\approx 3.0$, $a\approx 0.060$ and $c_0 \approx 0.23$.

We then ran CBPF for 9,000 iterations and 1,000 burn-in with the above estimated parameters to  infer both  $C_t$ and $N_t$. For the estimated $N_t$, and its reference value given in the dataset, we show the running sum over a one second window. Figure \ref{fig:cai3} shows the results of the inference: the 90\% credible interval for $C_t$ and the posterior mean of $N_t$. The inferred  $\mathrm{Ca}^{2+}$ follows  the data closely, and the mean of $N_t$ appears to align well with the reference values. We also show the normal quantiles of the errors $\hat{\epsilon}_{F,t} = (F_t-C_t^{(n)})/\sigma_{F}$, where $C_{1:T}^{(n)}$ is the final CBPF sample of the smoothing trajectory, suggesting a good fit.

Figure \ref{fig:cai3} (bottom) shows the per time index `acceptance rates,' that is, the rate when the reference changed at that time index. The rate is high in most places, but indicates some `stickiness' near the jumps ($N_t\ge 1$). Such behaviour could likely be rectified by using a more clever proposal and/or a more flexible model for jumps. 
\begin{figure}
  \includegraphics[width=\linewidth]{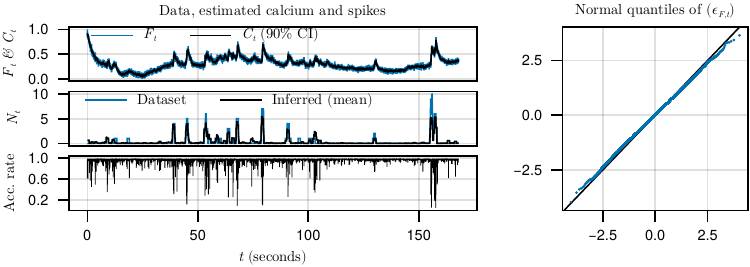}
  \caption{The (rescaled) fluorescence data and estimated $\mathrm{Ca}^{2+}$ concentration (top left); the supplied reference spikes and estimated $N_t$ (mean) within one second window (middle left); the rate of CBPF's reference changes (bottom left); and (right) the normal quantile plot for residuals (see Section \ref{sec:fluorescence}).}
  \label{fig:cai3}
\end{figure}
We also ran the iterated \textsc{CoupledCBPF} with IMC and investigated the coupling times $\tau$. Figure \ref{fig:cai3-couplings} (left) shows the distribution of 1000 independent realisations of $\tau$. Although the overall coupling time $\tau$ matters, we also investigated  the coupling times of each individual time instant, that is, $\tau_t = \inf\{n\ge 1\given [\mathbf{S}_k]_t=[\tilde{\mathbf{S}}_k]_t \;\forall k\ge n\}$. The distributions for $\tau_t$ are illustrated in Figure \ref{fig:cai3-couplings} for three selected segments. Such plots could potentially be useful when tuning the model and the proposals in the CBPF.
\begin{figure}
  \includegraphics[width=\linewidth]{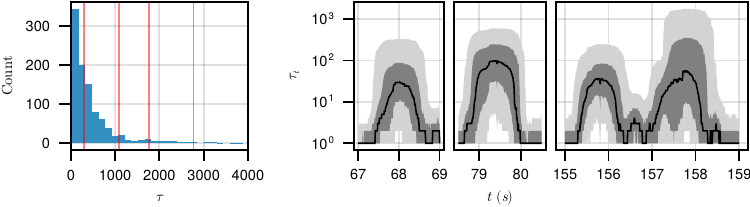}
  \caption{Figure on the left shows histogram of the (overall) coupling time of 1000 repeated runs of IMC with the model of Section \ref{sec:fluorescence}. The vertical red lines indicate the 50\%, 90\%, 95\% and 99\% quantiles. The three
  plots on the right illustrate the median, 50\% and 90\% intervals of per-time-index final coupling time around three `stickiest' segments.}
  \label{fig:cai3-couplings}
\end{figure}

\section{Discussion}
\label{sec:discussion}

We introduced a new coupling of the CBPF, which is guaranteed to admit a coupling time that increases only logarithmically in the time horizon length, if the number of particles is large enough. As a direct corollary, the mixing time of CBPF must behave similarly, which consolidates the empirical evidence that CBPF scales extremely well for long time horizons \citep[e.g.][]{lindsten-jordan-schon,lindsten-schon,chopin-papaspiliopoulos}. The coupling, and a number of modified couplings, are implementable and can be used for unbiased estimation.

Our theoretical results rely on the strong mixing assumption \ref{a:strong-mixing}, and the `sufficient' number of particles $N_{\rm min}$ which guarantees $r_N>1$ in Theorem \ref{thm:cbpf-mixing} depends on the strong mixing constants $0 <\Mlo\le \Mhi <\infty$, $0 < \Glo \le \Ghi<\infty$. Our experiments suggest that $N$ must indeed be `large enough', and $N_{\rm min}$ does need to grow when the strong mixing constants get worse. 

In many applications $\X=\R^d$ and \ref{a:strong-mixing} is not satisfied, like in our linear-Gaussian and stochastic volatility experiments, yet we still see qualitative behaviour consistent with our theory. If the potentials still satisfy the upper bound $\Ghi<\infty$, we know that the CPF kernel, as shown in \citep{andrieu-lee-vihola}, and the CBPF kernel (see \eqref{eq:andrieu-lee-vihola-bound}) remain uniformly ergodic. It is a natural question
whether a similar logarithmic scaling behaviour can hold in such a setting, and whether our theoretical approach could be extended to cover alternative mixing conditions. For instance, an extension might be based on conditions which guarantee a multiplicative drift \citep[cf.][]{douc-fort-moulines-priouret,whiteley-stability}, and/or argue that `typical' observed values of the potentials/transition densities still satisfy \ref{a:strong-mixing}.

Another question, related to the point above, is how mixing time and/or $N_{\rm min}$ evolve in increasing state dimension $d$? Intuitively, the `typical' variability of the values of the potentials/transition densities increases in $d$, and therefore a sufficient $N_{\rm min}$ might quickly become very large. Recently, an extension of the CPF based on Markov chain Monte Carlo transitions within the state updates was investigated in \citep{finke-thiery}. A natural direction of research is to investigate whether our techniques could be extended to cover their algorithm, which can be applied in higher dimensional state spaces.

Our strong mixing assumption \ref{a:strong-mixing} only allows for potentials $G_t(x_t)$ that depend on the current state $x_t$. The earlier result of \citep{lee-singh-vihola} allowed for dependence on the previous state, too, having $G_t(x_{t-1},x_t)$. We suspect that the theoretical results of the present paper could hold also in this setting, that is, with the generalised algorithm discussed in Appendix \ref{app:pairpotcoup}, but the extension is non-trivial.

CBPF (Algorithm \ref{alg:cbpf}) is valid also with resampling strategies other than multinomial \citep{chopin-singh,karppinen-singh-vihola}. When other resamplings are used, the filter's predictive distribution is no longer product form, but still admits the same one-step predictive distributions $\fsp_{t}$ as its marginals. In this case, coupling is still possible using Algorithm \ref{alg:mc_cbpf}, if \textsc{FwdCoupling} implements a coupling of the related laws. It may be possible to extend our theoretical analysis for such a case, and also to the bridge backward sampling generalisation of the CBPF \citep{karppinen-singh-vihola}. 

The blocked CPF introduced in \citep{singh-lindsten-moulines} is a variant of CPF/CBPF, which is applied to blocks of indices at a time, and which enjoys parallelisation benefits. The blocked CPF satisfies a similar mixing behaviour as Theorem \ref{thm:cbpf-mixing} guarantees for the CBPF, as long as the block length, block overlap and number of particles $N$ are chosen appropriately. The CBPF can be simpler to apply, because it only requires $N$ to be chosen, and the explicit coupling construction can also be used for unbiased estimation. The CBPF can also be used as a within-block sampler in the blocked CPF. Our result applied in this context suggest that $O(\log \tau)$ iterations for a block of length $\tau$ leads to a near-perfect Gibbs update, which could potentially be a practically useful strategy.

An alternative method for approximating smoothing expectations is to use variations of the forward-filtering backward-smoothing (FFBS) algorithm, which also involve running a particle filter and then using the transition densities (or a sophisticated coupling strategy) to approximate the smoothing distribution \citep[see, e.g.,][and references therein]{olsson-westerborn,dau_chopin_smoothing}.
The approach here is quite different as those algorithms are consistent as $N \to \infty$, but biased for finite $N$, whereas the CBPF leaves the smoothing distribution invariant and so arbitrarily accurate approximations can be obtained with fixed $N$ by iterating the CBPF more times. It is also possible to use the CBPF state variables
within an FFBS style estimator, which is unbiased in stationarity \citep{cardoso2023state}. Our result is directly relevant to the analysis of such methods.

In our experiments, the `independent maximal coupling' (IMC) algorithm was equally efficient, and sometimes notably more efficient, than the `joint maximal coupling' (JMC), which we analysed theoretically. We believe that IMC enjoys similar guarantees to JMC, but extending our results to cover IMC is not straightforward, and therefore its analysis remains an open question.

\section*{Acknowledgements}

JK and MV were supported by Research Council of Finland grant `FiRST (Finnish Centre of Excellence in Randomness and Structures)' (346311, 364216).
AL was supported by the Engineering and Physical Sciences Research Council (EP/R034710/1, EP/Y028783/1).
SSS holds the Tibra Foundation professorial chair and gratefully acknowledges research funding as
follows: This material is based upon work supported by the Air Force Office of
Scientific Research under award number FA2386-23-1-4100. 
The authors wish to acknowledge CSC --- IT Center for Science, Finland, for computational resources.


\appendix

\section{Properties of the CBPF forward process}
\label{app:cpf-theory}

  This appendix is devoted to properties of the forward process of the CBPF, that is, lines  \ref{line:cbpf-forward-start}--\ref{line:cbpf-forward-end} of Algorithm \ref{alg:cbpf}. The main results that pertain to these lines are lemmas  \ref{lem:fullcoupling-tv} and \ref{lem:tvlemma_new}.
  
  We use the common notation for a probability measure $\mu$, real-valued function $f$ and Markov kernel $M$: we write $\mu(f) = \int f(x) \mu (\ud x)$, and $(\mu M)(A) = \int \mu(\ud x) M(x,A)$. The composition of two Markov kernels $M_1$ and $M_2$ is $(M_1 M_2)(x,A) = \int M_1(x,\ud y) M_2(y,A)$. The oscillation of a function is denoted as $\osc(f) = \sup_x f(x) - \inf_x f(x)$, and $\| X \|_p = (\E|X|^p)^{1/p}$ stands for the $L^p$-norm of a random variable $X$. The Markov kernels corresponding to the HMM/Feynman--Kac transition densities $M_t$ are denoted by the same symbols.
  
  The Hellinger distance $\HD(\mu,\nu)$ between two distributions $\mu$ and $\nu$ is defined via
  $$
     \HD^2(\mu,\nu) = \frac{1}{2}\int \big(\sqrt{p(x)}-\sqrt{q(x)}\big)^2 \lambda(\ud x) = 1 - \int \sqrt{p(x)q(x)} \lambda(\ud x),
  $$
  where $p = \ud \mu/\ud \lambda$ and $q = \ud \nu/\ud \lambda$ are densities of $\mu$ and $\nu$ with respect to some common dominating measure $\lambda$ (which always exists, and the definition is invariant to the choice).
  We may upper bound the expected total variation distance of product measures by their expected Hellinger distance:
  \begin{lemma}
    \label{lem:dimfree-new}
    Let $\mu$ and $\nu$ be random
    probability measures. Then, for any $N \geq 1$:
    \begin{align*}
      \E\big\Vert   \mu^{\otimes N} - \nu^{\otimes N} \big\Vert _{\mathrm{TV}} 
     & \leq\sqrt{1-\left(1-\E [\HD^{2}(\mu, \nu)] \right)^{2N}}.
    \end{align*}
    \end{lemma}
  \begin{proof}
  The proof follows from Le Cam's inequality \citep[e.g.][Lemma 28]{karjalainen-lee-singh-vihola}:
  $$
  \| \mu^{\otimes N} - \nu^{\otimes N} \|_{\mathrm{TV}}^2 \le
  1 - \big(1- \HD^2(\mu^{\otimes N}, \nu^{\otimes N})\big)^2,
  $$
  the fact that $1- \HD^2(\mu^{\otimes N}, \nu^{\otimes N}) = \big(1-\HD^2(\mu,\nu)\big)^N$ and Jensen's inequality.
  \end{proof}
  
  The following result is a restatement of \citep[Lemma 29]{karjalainen-lee-singh-vihola}, and states that expected Hellinger distance may be upper bounded by $L^2$ `errors':
  \begin{lemma}
    \label{lem:h2bound}
    Let $\mu$ and $\nu$ be two random probability measures. If \textup{\ref{a:strong-mixing} (M)} holds for a Markov kernel $M$, then
    \[
    \E \HD^2(\mu M, \nu M) \leq c'  \sup_{\osc(\phi) \leq 1} \E (| \mu(\phi) - \nu(\phi) |^2),
    \]
    where $c' = \frac{1}{8}(\Mhi / \Mlo)^2$.
  \end{lemma}

  We need the following operators associated with the Feynman--Kac model, which map probability measures $\mu$ to probability measures:
  \begin{align*}
    \Psi_t(\mu)(\ud x) &= \frac{G_t(x) \mu(\ud x)}{\mu(G_t)}, & t&\in\{1{:}T\}, \\
    \Phi_{t}(\mu) &= \Psi_{t-1}(\mu)M_{t}, & t&\in\{2{:}T\}, 
  \end{align*}
  and the compositions of updates from time $t$ to time $u$ as:
  \begin{align*}
    \Phi_{t,u} &= \Phi_u \circ \cdots \circ \Phi_{t+1}, & 1 &\leq t < u\le T, 
  \end{align*}
  and $\Phi_{t,t}$ stands for the identity operator.
  In the engineering and statistics literature, $\eta_t = \Phi_{1,t}(M_1)$ is known as the (ideal) `predictor'.  The following result, which is a restatement of \citep[Proposition 4.3.6]{del2004feynman}, shows that the ideal predictor forgets its initial distribution exponentially fast.
  
  \begin{lemma}
    \label{lem:fkcontract}
    Assume \textup{\ref{a:strong-mixing}}. For all probability measures $\mu$ and $\nu$ on $\X$, $1 \le t \le T$ and $0 \le k\le T-t$:
    \[
     \sup_{\mu, \nu} \| \Phi_{t,t+k}(\mu) - \Phi_{t,t+k}(\nu) \|_{\mathrm{TV}} \leq \beta^{k},
     \]
  where $\beta = 1-(\underline M / \bar M)^2$.
  \end{lemma}
  
  We also need the following well-known fact \citep[e.g.][Lemma 8]{karjalainen-lee-singh-vihola}, that the weighting operator can only inflate the $L^p$ errors by a constant factor.
  
  \begin{lemma}[\cite{karjalainen-lee-singh-vihola}, Lemma 8]
  \label{lem:bgtransformlp}
  Let $\mu$ and $\nu$ be random probability measures and $p\geq 1$. If \textup{\ref{a:strong-mixing} (G)} holds, then for all $1\le t\le T$:
  \[
  \sup_{\osc(\phi) \leq 1} \|\Psi_t(\mu)(\phi) - \Psi_t(\nu)(\phi) \|_p \leq c \sup_{\osc(\phi) \leq 1} \|\mu(\phi) - \nu(\phi) \|_p,
  \]
  where $c= \Ghi / \Glo$.
  \end{lemma}

  The following `perturbed' analogues of the above are associated with the CPF:
  \begin{align*}
    \Psi_t^{x_{t}^*} (\mu) &= \Psi_t \left( \frac{1}{N+1} \delta_{x_{t}^*} + \frac{N}{N+1}\mu \right),\\
    \Phi_t^{x_{t-1}^*} (\mu) &= \Phi_t \left( \frac{1}{N+1} \delta_{x_{t-1}^*} + \frac{N}{N+1}\mu \right). 
  \end{align*}
  Note that the one-step predictive distributions defined in \eqref{eq:filter-state-predictive} may be written as $\fsp_t^N = \Phi_t^{x_{t-1}^*}(\eta_{t-1}^N)$ and $\tilde{\fsp}_t^N = \Phi_t^{\tilde{x}_{t-1}^*}(\tilde{\eta}_{t-1}^N)$, where $\eta_{t-1}^N$ and $\tilde{\eta}_{t-1}^N$ stand for the empirical measures of $X_{t-1}^{1:N}$ and $\tilde{X}_{t-1}^{1:N}$, respectively. 
  
  The following restatement of \citep[Theorem 21]{karjalainen-lee-singh-vihola} indicates that the perturbed and weighted empirical measures $\Psi_t^{x_t^*}(\eta_t^N)$ from the CPF approximate the `ideal filter' $\Psi_t(\eta_t)$ in $L^p$ sense, uniformly in time. 
  \begin{theorem}
  \label{thm:cpfstability}
  Assume \textup{\ref{a:strong-mixing}}. For every $p \geq 1$, there exists a constant $c=c(p)$ such that for all $\phi$ with $\osc(\phi) \leq 1$, all $1\le t \le T$, $N\ge 1$ and all references $x_{1:T}^*\in\X^T$:
  \begin{align*}
  \left\|\Psi_t^{x_{t}^*} (\eta_{t}^N) (\phi) - \Psi_t(\eta_{t})(\phi)\right\|_p \leq  &\, \frac{c}{\sqrt{N}}.
  \end{align*}
  In particular, one may choose
  \begin{equation}
  \label{eq:cpfconstant}
  c = \frac{\Ghi}{ 2 \Glo} + \left( \frac{\Ghi}{\Glo} \right)^2 \left( \frac{\Mhi}{\Mlo} \right)^3 \left( d(p)^{1/p} + \frac{1}{2} \right),
  \end{equation}
  where the function $d(p)$ has been defined in \citep{del2004feynman}, and in particular, $d(2) = 1$.
  \end{theorem}
  
  We may write the Markov transitions corresponding to one step of the CPF forward process (line \ref{line:mixture-transition} of Algorithm \ref{alg:cbpf}):
  \[
  \M_t^{x_{t-1}^*}(x^{1:N}, \uarg) 
  = \left(  \Phi_t^{x_{t-1}^*}\left( \frac{1}{N} \sum_{i=1}^N \delta_{x^i} \right) \right)^{\otimes N}.
  \]
  We denote the compositions of the above as
  $$
  \M_{t,t+j}^{x_{t:t+j-1}^*} = 
  \M_{t+1}^{x_{t}^*} \M_{t+2}^{x_{t+1}^*} \ldots \M_{t+j}^{x_{t+j-1}^*}.
  $$
  
  Our first result on the coupled CPF (Algorithm \ref{alg:mc_cbpf}) is an upper bound of the total variation between products of predictive distributions, which are coupled in \textsc{FwdCoupling}, when all previous particles are coupled:
  \begin{lemma}
    \label{lem:fullcoupling-tv}
  Under \textup{\ref{a:strong-mixing}}, whenever $X_{t-1}^{1:N} = \tilde{X}_{t-1}^{1:N}$, we have:
    $$
    \tv{(\fsp_t^N)^{\otimes N} - (\tilde{\fsp}_t^N)^{\otimes N}} \le \epsilon_N^\mathrm{uc}, 
    \quad\text{where}\quad   \epsilon_N^\mathrm{uc} = \frac{\Mhi\Ghi}{2\Mlo\Glo} \frac{1}{\sqrt{N+1}}.
    $$
  \end{lemma}
  From the definition of $\zeta_{t}^{N}$and $\tilde{\zeta}_{t}^{N}$
  in \eqref{eq:filter-state-predictive}, when $X_{t-1}^{1:N}=\text{\ensuremath{\tilde{X}_{t-1}^{1:N}}}$,
  these two empirical measures can only differ in particle zero, in both
  its value and un-normalised weight. The stated result holds irrespective
  of $(X_{t-1}^{0},\tilde{X}_{t-1}^{0})=(x_{t-1}^{\ast},\tilde{x}_{t-1}^{\ast})$.
  \begin{proof}[Proof of Lemma \ref{lem:fullcoupling-tv}]
  By (a deterministic version of) Lemma \ref{lem:dimfree-new}, it holds for arbitrary $X_{t-1}^{1:N}$, $\tilde{X}_{t-1}^{1:N}$ that
  $$
  \tv{(\fsp_t^N)^{\otimes N} - (\tilde{\fsp}_t^N)^{\otimes N}} \le
  \sqrt{1-(1-\HD^2(\fsp_t^N, \tilde{\fsp}_t^N))^{2N}},
  $$
  and by Lemma \ref{lem:h2bound}, 
  $$
  \HD^2(\fsp_t^N, \tilde{\fsp}_t^N)\leq \frac{1}{8}(\Mhi/\Mlo)^2 \sup_{{\rm osc}(\phi)\leq 1}|\xi(\phi) - \tilde{\xi}(\phi)|^2,
  $$
  where $\xi = {\Psi}_{t-1}^{x_{t-1}^{*}}(\eta_{t-1}^N)$ and 
  $\tilde{\xi}={\Psi}_{t-1}^{\tilde{x}_{t-1}^*}(\tilde{\eta}_{t-1}^N)$. 
  By Lemma \ref{lem:bgtransformlp}, 
  \begin{align*}
  |\xi(\phi) - \tilde{\xi}(\phi)| &\leq \frac{\Ghi}{\Glo} \frac{1}{N+1} \sup_{\osc(\phi) \leq 1} \left|\left( \sum_{i=1}^N \delta_{X_{t-1}^i} + \delta_{x_{t-1}^*} \right)\phi - \left( \sum_{i=1}^N \delta_{\tilde X_{t-1}^i} + \delta_{\tilde x_{t-1}^*} \right)\phi \right| \\
  &= \frac{\Ghi}{\Glo} \frac{1}{N+1},
  \end{align*}
  where the last equality follows because of our assumption $X_{t-1}^{1:N}=\tilde{X}_{t-1}^{1:N}$.
  
  Combining these results gives the upper bound
  \begin{align*}
    \tv{(\fsp_t^N)^{\otimes N} - (\tilde{\fsp}_t^N)^{\otimes N}} 
  & \leq 
  \sqrt{1-\left(1-\frac{1}{8} \frac{\Mhi^2}{\Mlo^2}\frac{\Ghi^2}{\Glo^2} \frac{1}{(N+1)^2}\right)^{2N}}.
  \end{align*}  
  Denote $D=(1/8)( \Mhi/\Mlo)^2 (\Ghi/\Glo)^2$. Then Bernoulli's inequality $(1+x)^r \geq 1+rx$ for all $x\geq -1, r \geq 1$, yields, whenever $(N+1)^2 \ge D$
  \[
  \left(1-\frac{D}{(N+1)^2}\right)^{2N} \geq 1-\frac{2ND}{(N+1)^2},
  \]
  and so
  \[
  \sqrt{1-\left( 1-\frac{D}{(N+1)^2}\right)^{2N}} \leq \sqrt{\frac{2ND}{(N+1)^2}} \leq \sqrt{\frac{2D}{N+1}}.
  \]
  If $(N+1)^2 < D$, then $\epsilon_N^\mathrm{uc} \ge 1$ and the claim holds trivially. 
  \end{proof}
  
  The following result 
  is slightly stronger than the forgetting result for the CPF established \citep[Lemma 19]{karjalainen-lee-singh-vihola}, but the proof remains similar.
  
  \begin{lemma}
    \label{lem:tvlemma_new}
    There exist $c$ and $c'$, only depending on the constants in \textup{\ref{a:strong-mixing}}, such that for all 
    $N \geq c'$, $j \geq c \log(N)$, $(x^{1:N}, \tilde x^{1:N})$, $(x_{1:T}^\ast, \tilde x_{1:T}^\ast)$ and $j+1 \le t \le T-1$,
    \[
      \sup_{\gamma\in\Gamma_{t-j,t}(x^{1:N}, \tilde{x}^{1:N})}
    \E_\gamma \tv{\M_{t+1}^{x_{t}^*}(X^{1:N}_{t}, \uarg)-\M_{t+1}^{\tilde x_{t}^*}(\tilde{X}^{1:N}_{t}, \uarg) } \leq (1-\varepsilon^2)^{1/2},
    \]
    where $\varepsilon = (1-c'/N)^N$ and where $\Gamma_{t-j,t}(x^{1:N}, \tilde{x}^{1:N})$ stands for all couplings of 
    $(X_t^{1:N}, \tilde{X}_t^{1:N})$ which have marginal laws $X_t^{1:N}\sim \M_{t-j,t}^{x_{t-j:t-1}^*}(x^{1:N}, \uarg)$ and $\tilde{X}_{t}^{1:N} \sim \M_{t-j,t}^{\tilde{x}_{t-j:t-1}^*}(\tilde{x}^{1:N}, \uarg)$.
  \end{lemma}
  
  \begin{proof}
    Let $\gamma \in \Gamma_{t-j,t}(x^{1:N}, \tilde{x}^{1:N})$ be arbitrary and $(X_t^{1:N}, \tilde{X}_t^{1:N}) \sim \gamma$.
    Define $ \xi_t^N =  N^{-1} \sum_{i=1}^N \delta_{X_t^i}$ and $\tilde \xi_t^N =  N^{-1} \sum_{i=1}^N \delta_{\tilde X_t^i}$
    and let
    \[\mu_{t+1}={\Psi}_{t}^{x_{t}^\ast }(\xi_t^N) M_{t+1}, \qquad \tilde{\mu}_{t+1}={\Psi}_{t}^{\tilde{x}_{t}^\ast }(\tilde{\xi}_t^N)M_{t+1}.
    \]
    Then, $\mu_{t+1}^{\otimes N} = \M_{t+1}^{x_{t}^*}(X^{1:N}_{t},\uarg)$ and 
    $\tilde{\mu}_{t+1}^{\otimes N} = \M_{t+1}^{\tilde{x}_{t}^*}(\tilde{X}^{1:N}_{t},\uarg)$,
    and by Lemma \ref{lem:dimfree-new},
    \begin{align*}
    \E \big\Vert  \mu_{t+1}^{\otimes N} - \tilde \mu_{t+1}^{\otimes N} \big\Vert _{\mathrm{TV}} 
    \leq \sqrt{1-\left(1-\E \HD^{2}(\mu_{t+1}, \tilde \mu_{t+1}) \right)^{2N}}.
    \end{align*}
    Denote $\HD^2_{t+1} = \HD^{2}(\mu_{t+1}, \tilde \mu_{t+1})$. Applying Lemma \ref{lem:h2bound} to $\mu_{t+1}$ and $\tilde{\mu}_{t+1}$  yields
    \[
    \E \HD^2_{t+1} \leq C \sup_{\osc(\phi) \leq 1} \E \big[| {\Psi}_{t}^{x_{t}^\ast }(\xi_t^N)(\phi) - {\Psi}_{t}^{\tilde{x}_{t}^\ast }(\tilde \xi_t^N)(\phi)|^2\big],
    \]
    where $C = (1/8) (\Mhi/ \Mlo)^2$. For $j\geq 1$, define
    \[
    \xi_{t-j,t} = \Phi_{t-j+1,t}\left ( \Phi_{t-j+1}^{x_{t-j}^*}\bigg(\frac{1}{N} \sum_{i=1}^N \delta_{x^i} \bigg) \right ), \; 
    \tilde \xi_{t-j,t} = \Phi_{t-j+1,t}\left ( \Phi_{t-j+1}^{\tilde{x}_{t-j}^*}\bigg(\frac{1}{N} \sum_{i=1}^N \delta_{\tilde{x}^i} \bigg) \right ).
    \] 
    Then,
    \begin{align}
    \|  \Psi_{t}^{x_{t}^\ast } (\xi_t^N)(\phi) - \Psi_{t}^{\tilde x_{t}^\ast } (\tilde  \xi_t^N)(\phi) \|_2 &\leq \| \Psi_{t}^{x_{t}^\ast } ( \xi_t^N)(\phi) - \Psi_{t}(\xi_{t-j,t}) (\phi) \|_2 \label{eq:filter-triangle} \\
    & \,+ |\Psi_{t}(\xi_{t-j,t}) (\phi) - \Psi_{t}(\tilde \xi_{t-j,t}) (\phi) | \nonumber\\
    & \,+ \| \Psi_{t}^{\tilde x_{t}^\ast }(\tilde{\xi}_t^N)(\phi) - \Psi_{t}(\tilde \xi_{t-j,t})(\phi) \|_2.\nonumber
    \end{align}
    Note that $X_t^{1:N}$ has the law of the particles in Algorithm \ref{alg:cbpf} using the initial distribution $ \Phi_{t-j+1}^{x_{t-j}^*}\big(N^{-1} \sum_{i=1}^N \delta_{x^i} \big)$, potential functions $G_{t-j+1:t-1}$, references $x^*_{t-j+1:t-1}$, and Markov kernels $M_{t-j+2:t}$. Thus, Theorem \ref{thm:cpfstability} gives the bound
    \[
      \| \Psi_{t}^{x_{t}^\ast } ( \xi_t^N)(\phi) - \Psi_{t}(\xi_{t-j,t}) (\phi) \|_2 
      \leq \frac{C_2}{\sqrt N},
    \]
    and by the same argument for $\tilde \xi_t^N$, the third term in \eqref{eq:filter-triangle} can be upper bounded by $C_2/\sqrt{N}$.
    Lemma \ref{lem:bgtransformlp} (with $p=1$) and Theorem \ref{lem:fkcontract} together give
    \[
      |\Psi_{t}(\xi_{t-j,t}) (\phi) - \Psi_{t}(\tilde \xi_{t-j,t}) (\phi) | \leq C_3 \beta^{j}
    \]
    with $C_3 = \Ghi / \Glo$ and $\beta = 1-(\Mlo /\Mhi)^2$. 
    If $ j \geq c \log(N)$ with 
    \[
    c = \frac{1}{2 \log (\beta^{-1})},
    \]
    then $\beta^j \leq N^{-1/2}$ for all $N\geq 2$.   
    Since $C_2 \geq C_3$ by \eqref{eq:cpfconstant}, it follows that 
    \[
    \|  \Psi_{t}^{x_{t}^\ast } (\xi_t^N)(\phi) - \Psi_{t}^{\tilde x_{t}^\ast } (\tilde  \xi_t^N)(\phi) \|_2
    \leq 3\frac{C_2}{\sqrt N},
    \] 
    and we conclude that $\E \HD^2_{t+1}  \leq c' N^{-1}$ with $c' = 9C C_2^2$.
\end{proof}

\section{Contraction for the number of uncoupled states}
\label{app:contraction}

This appendix is devoted to the proof of Theorem \ref{thm:main-contraction}, which states that the number of uncoupled elements in the reference paths is shrunk by a factor $\lambda_N<1$ for sufficiently large $N\ge 1$ under \ref{a:strong-mixing} and \ref{a:joint-maximal-coupling}.
We study the random variables generated within Algorithm \ref{alg:mc_cbpf}, for arbitrary references $x_{1:T}^*$ and $\tilde{x}_{1:T}^*$, which are regarded as fixed for the rest of this section.
In what follows, we denote $[T] = \{1,\ldots,T\}$.

Let us first define the distance (in time) from each time index $t\in[T]$ to the nearest uncoupled state:
$$
   d_t = \inf\{ | t -u | \given u\in [T]
   ,\, x_u^*\neq \tilde{x}_u^*\}.
$$
\begin{figure}
  \begin{center}
    \includegraphics[scale=0.9]{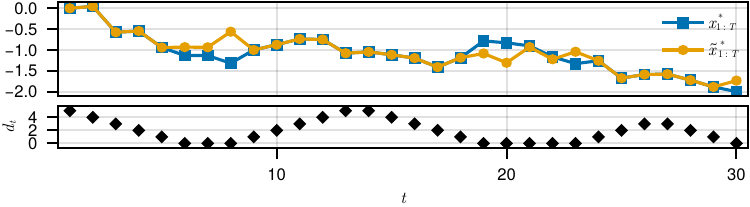}
  \end{center}
  \caption{Partially coupled reference trajectories and distances $d_t$ to nearest uncoupled states.}
  \label{fig:reference-distances}
\end{figure}
Then, let $H^d\subset [T]$ stand for the time indices with distance $d$ to the nearest uncoupled site (see Figure \ref{fig:reference-distances}):
$$
   H^d = \{t \in [T]\given d_t = d \}.
$$
This means that $H^0 = \{t \in [T]\given x_t^* \neq \tilde{x}_t^*\}$ are the uncoupled states (`holes' in coupling), and for $d\ge 1$, any time index $t\in H^d$ has the nearest hole either at $t-d$ or $t+d$.

Before proceeding with the proof, we describe informally the phenomenon being analysed. By keeping track of the coupled and uncoupled states when simulating several iterations of the coupled CBPF under \ref{a:joint-maximal-coupling}, one can see that sets of coupled states do not only `grow' from the left as with the different forward coupling (see Section \ref{sec:ind-ix-coupling}) studied in \citep{lee-singh-vihola} but can appear spontaneously in isolation and tend to grow as iterations progress. This behaviour was also observed for the much simpler coupled Gibbs sampler for an AR(1) process by \cite{wilkinson-discussion}, and is essential to understanding the $O(\log T)$ mixing time. In order to capture this phenomenon, we will first show that the forward coupling process ensures that particle filter states have some probability of being coupled at each time, and this increases for time indices far from a hole. This ultimately relies on the time-uniform stability of the particle filter and properties of maximal couplings. Using these results, we will show that the output states after backward sampling are also likely to be coupled, especially for time indices far from a hole. This then leads to the conclusion that, for large enough $N$, output states can be spontaneously coupled and progressively consolidated as iterations progress, which is rigorously captured by the multiplicative reduction in the expected number of holes in each iteration.

The proof of Theorem \ref{thm:main-contraction} will be based on the following easy observation:
\begin{lemma}
  \label{lem:number-to-individual-sites}
  In the setting of Theorem \ref{thm:main-contraction},
  $$
  \E[B] \le b^* \bigg( \sup_{t\in H^0} \P(X_t \neq \tilde{X}_t) + 2 \sum_{d=1}^\infty 
  \sup_{t \in H^d} \P(X_t \neq \tilde{X}_t) \bigg).
  $$
\end{lemma}
\begin{proof}
If $b^*=\sum_{t=1}^T \I(x_t^*\neq \tilde{x}_t^*)=0$, then it is easy to see that all couplings in Algorithm \ref{alg:mc_cbpf} happen with probability one, so $X_t^i=\tilde{X}_t^i$ and $B=\sum_{t=1}^T \I(X_t\neq \tilde{X}_t)=0$. Otherwise, $\{H^d\}_{d\ge 0}$ form a partition of $[T]$, and we may write
\begin{align*}
       \E[B]
       & = \E\bigg[\sum_{d\ge 0} \sum_{t\in H^d} \I(X_t \neq \tilde{X}_t)\bigg] \\
       & \le |H^0| \sup_{t\in H^0} \P(X_t \neq \tilde{X}_t) + \sum_{d=1}^\infty |H^d|  \sup_{t\in H^d} \P(X_t \neq \tilde{X}_t),
\end{align*}
from which the claim follows, because 
$b^*= |H^0|$ and $|H^d| \le 2 b^*$ for $d\ge 1$.
\end{proof}
A few comments on the rationale of the chosen upper bound
in the main statement of Lemma \ref{lem:number-to-individual-sites} are in order. Firstly, the largest $d$
for which $H^{d}$ can be non-empty is $T-1$; for example, when $x_{1}^{\ast}\neq\tilde{x}_{1}^{\ast}$
but $x_{2:T}^{\ast}=\tilde{x}_{2:T}^{\ast}$. Summing over larger
values of $d$ is geared towards a result that holds for any $T$.
Secondly, the bound (before application of $\sup_t$) can be accurate when the set of times $t$ such that $x_{t}^{\ast}\neq\tilde{x}_{t}^{\ast}$
is sparse. For example, in the case of a single unequal pair
of reference states at the trajectory mid-point, 
the chosen bound is appropriate.

In order to control the probabilities in the upper bound of Lemma \ref{lem:number-to-individual-sites}, we need to analyse both the behaviour of the forward coupling process (lines \ref{line:forward-start}--\ref{line:forward-end} of Algorithm \ref{alg:mc_cbpf}) and the backward index coupling (lines \ref{line:backward-start}--\ref{line:backward-end} of Algorithm \ref{alg:mc_cbpf}). 

\subsection{Forward coupling process}

The first part of our analysis concerns the behaviour of the forward pass of Algorithm \ref{alg:mc_cbpf}. To that end, denote the event of `full coupling' by $F_t = \{X_t^{1:N} = \tilde{X}_t^{1:N}\}$, that is, all (non-reference) particles are equal at $t$. This means that \textsc{FwdCouple} successfully coupled the one-step particle state predictive distributions. Likewise, denote by $E_t$ the complement of $F_t$. We use the shorthand 
$F_{t:u} = \{ F_t, F_{t+1},\ldots, F_u\}$ for $1 \le u \le t \le T$, and similarly for intersections of $E_t$.

There are two main probabilities we strive to characterise in this
section. In particular, we will derive the bounds $\mathbb{P}(E_{t})
=O(N^{-1/2})$ for $t\in H^{0}$ and $\mathbb{P}(E_{t}) = O(N^{-1/2}) \exp\big(-\Omega(\frac{d}{\log N})\big)$ for $t\in H^{d}$ with $d\ge 1$ (Lemma \ref{lem:fwdcoupling-in-Hd}). 
These are time uniform bounds that decay with more particles $N$, and the latter bound also vanishes quickly in the `depth' $d$. That is, for the time points $t\in H^d$ that are deep within a block of contiguous coupled reference states, the probability $\P(E_t)$ is negligible.

  Our first lemma states that if we have full coupling at time $t$, we are likely to preserve it at $t+1$. 
\begin{lemma}
  \label{lem:uncoupling}
  $\P(F_1)=1$ and for $t=1,\ldots,T-1$, whenever $F_{t}$ holds:
  \begin{enumerate}[(i)~]
    \item \label{item:unlikely-uncoupling} 
    $\P(E_{t+1} \mid X_t^{1:N},\tilde{X}_t^{1:N}) \le \epsilon_N^\mathrm{uc} \wedge 1$, and 
    \item \label{item:impossible-uncoupling} $\P(E_{t+1} \mid X_t^{1:N},\tilde{X}_t^{1:N}) = 0$ if also $x_t^{*} = \tilde{x}_t^{*}$,
  \end{enumerate}
  where $\epsilon_N^\mathrm{uc}$ is given in Lemma \ref{lem:fullcoupling-tv} in Appendix \ref{app:cpf-theory}. 
\end{lemma}
\begin{proof}
   The claim $\P(F_1)=1$ follows by definition; see line \ref{line:initial-coupling} of Algorithm \ref{alg:mc_cbpf}. Claim \eqref{item:impossible-uncoupling} follows from line \ref{line:maximalcoupling} because if $x_t^* = \tilde{x}_t^*$ and $F_t$ happens, the particles $X_{t+1}^{1:N}$ and $\tilde X_{t+1}^{1:N}$ are sampled from the same distribution, and we use maximal coupling \ref{a:joint-maximal-coupling}. For \eqref{item:unlikely-uncoupling}, we may write
   \begin{align*}
   \I(F_{t})\P(E_{t+1} \mid X_{t}^{1:N}, \tilde{X}_{t}^{1:N}) &= \I(F_{t})\big\| (\fsp_{t+1}^N)^{\otimes N} - (\tilde{\fsp}_{t+1}^N)^{\otimes N}\big\|_\mathrm{TV},
   \end{align*} 
   where $\fsp_{t+1}^N$ and $\tilde{\fsp}_{t+1}^N$ are defined in \eqref{eq:filter-state-predictive}. The upper bound follows by Lemma \ref{lem:fullcoupling-tv} in Appendix \ref{app:cpf-theory}.
\end{proof}

\begin{lemma}
\label{lem:tvuniform_ccpf}
There exist $N_{\rm min} \in[2,\infty)$, $c\in(1,\infty)$,  and $\alpha \in [0,1)$, only depending on the constants in \textup{\ref{a:strong-mixing}}, such that for all 
$t \geq 2$, $N \geq N_{\rm min}$, and $T-t \geq j \geq c \log(N)$, it holds that
\[
\P(E_{t+j} \mid X_{t-1}^{1:N},\tilde{X}_{t-1}^{1:N} ) \leq \alpha,\qquad \rm{almost~surely}.
\]
\end{lemma}
\begin{proof}
By the Markov property and the assumption that at time step $t+j$, $X_{t+j}^{1:N}$ and $X_{t+j}^{1:N}$ are generated using a maximal coupling, we have
\begin{align*}
\P &(E_{t+j}\mid X_{t-1}^{1:N}, \tilde{X}_{t-1}^{1:N}) \\
   &= \E\big[ \P(E_{t+j}\mid X_{t+j-1}^{1:N}, \tilde{X}_{t+j-1}^{1:N} ) \;\big|\; X_{t-1}^{1:N}, \tilde{X}_{t-1}^{1:N}\big] \\
&= \E\big[ \big\| \M_{t+j}^{x_{t+j-1}^*}(X^{1:N}_{t+j-1}, \uarg)-\M_{t+j}^{\tilde x_{t+j-1}^*}(\tilde{X}^{1:N}_{t+j-1}, \uarg) \big\|_\mathrm{TV} \big| \; X_{t-1}^{1:N}, \tilde{X}_{t-1}^{1:N}\big].
\end{align*}
By Lemma \ref{lem:tvlemma_new} in Appendix \ref{app:cpf-theory}, there exist $c$ and $c'$ only depending on the constants in \ref{a:strong-mixing} such that for all $N \geq c'$ and $j \geq c \log N$, 
\begin{align}
\P (E_{t+j}\mid X_{t-1}^{1:N}, \tilde{X}_{t-1}^{1:N}) \leq (1-\varepsilon^2)^{1/2},
\label{eq:alpha-N-expression}
\end{align}
where $\varepsilon = (1-c'/N)^N$. The function $N \mapsto (1-(1-c'/N)^{2N})^{1/2}$ is decreasing for $N > c'$. Thus, the claim follows by choosing $N_{\rm min} = \lfloor c' \rfloor +1$, and
\[
\alpha = \big(1-(1-c'/N_{\rm min})^{2 N_{\rm min}}\big)^{1/2}. \qedhere
\]
\end{proof}

\begin{remark}
  \label{rem:alpha}
  In the proof of Lemma \ref{lem:tvuniform_ccpf}, $\alpha$ depends on the specific choice of $N_{\rm min}$ 
  and can be made arbitrarily close to $\sqrt{1-\exp(-2c')}$ by increasing $N_{\rm min}$. 
We have opted to avoid using the bound exactly as in \eqref{eq:alpha-N-expression}, which depends on $N$ in a complicated way, to simplify the presentation.
\end{remark}

Our next result formalises the consequence of Lemma \ref{lem:tvuniform_ccpf} that full couplings should occur every $O(\log N)$ steps.
\begin{lemma}
  \label{lem:remain-uncoupled}
  Let $(N_{\rm min}, c, \alpha)$ be as in Lemma \ref{lem:tvuniform_ccpf}. For all $t\ge 2$, $N \geq N_{\rm min}$, and $1\le j\le T-t$, it holds that with $s_N = \lceil c\log N \rceil + 1$:
  $$
     \P(E_{t},E_{t+1},\ldots, E_{t+j}\mid X_{t-1}^{1:N},\tilde{X}_{t-1}^{1:N}) \le \alpha^{\lfloor \frac{j}{s_N}\rfloor},\qquad \rm{almost~surely}.
  $$
\end{lemma}
\begin{proof}
Let $c$ be as in Lemma \ref{lem:tvuniform_ccpf}. To avoid double subscripts in the proof, we denote $s = s_N$ 
and let $d = \lfloor j/s \rfloor$. By considering only the events which are $s$ steps apart, we obtain the trivial bound
\[
 \P(E_{t},E_{t+1},\ldots, E_{t+j}\mid X_{t-1}^{1:N},\tilde{X}_{t-1}^{1:N}) 
 \leq  \P\bigg(\bigcap_{i=1}^d E_{t-1+is} \;\bigg|\; X_{t-1}^{1:N},\tilde{X}_{t-1}^{1:N}\bigg).
\]
The claim is trivial for $d=0$, and for $d\ge 1$, the right-hand side can be written as
\begin{align*}
&\E\bigg[ \P\bigg(\bigcap_{i=1}^d E_{t-1+is} \;\bigg|\; X_{t-1:t-1+(d-1)s}^{1:N},\tilde{X}_{t-1:t-1+(d-1)s}^{1:N} \bigg) \;\bigg|\; X_{t-1}^{1:N},\tilde{X}_{t-1}^{1:N} \bigg] \\
&= \E\bigg[ \I\bigg(\bigcap_{i=1}^{d-1} E_{t-1+is} \bigg) \P\big( E_{t-1+ds} \;\big|\; X_{t-1+(d-1)s}^{1:N},\tilde{X}_{t-1+(d-1)s}^{1:N}\big) \;\bigg|\; X_{t-1}^{1:N},\tilde{X}_{t-1}^{1:N} \bigg].
\end{align*}
By Lemma \ref{lem:tvuniform_ccpf}, $\P\big( E_{t-1+ds} \;\big|\; X_{t-1+(d-1)s}^{1:N},\tilde{X}_{t-1+(d-1)s}^{1:N}\big) \le \alpha$, so we obtain
 \begin{align*}
  \P(E_{t},E_{t+1},\ldots, E_{t+j}\mid X_{t-1}^{1:N},\tilde{X}_{t-1}^{1:N})  \le
\alpha \P \bigg( \bigcap_{i=1}^{d-1} E_{t-1+is} \;\bigg|\; X_{t-1}^{1:N},\tilde{X}_{t-1}^{1:N} \bigg).
\end{align*}
Iterating the same argument $d-1$ times gives the claim.
\end{proof}

The three lemmas above are already sufficient to conclude that forward coupling will become very likely for large $N$, uniformly in time.

\begin{lemma}
  \label{lem:empty-unlikely}
  Let $(N_{\rm min}, c, \alpha)$ be as in Lemma \ref{lem:tvuniform_ccpf}.
  The probability of failing to couple can be upper bounded for all $t\in[T]$ and $N\ge N_{\rm min}$ as follows:
  $$
     \P(E_t) \le \epsilon_N^{E,0},\qquad\text{where}\qquad \epsilon_N^{E,0} = \epsilon_N^\mathrm{uc} \Big( 1 + \frac{ s_N }{1-\alpha} \Big),
  $$
  where $s_N = \lceil c \log(N)\rceil + 1$ and $\epsilon_N^\mathrm{uc}$ is as in Lemma \ref{lem:fullcoupling-tv} in Appendix \ref{app:cpf-theory}.
\end{lemma}
\begin{proof}
The claim is trivial for $t=1$, because $F_1$ always happens (Lemma \ref{lem:uncoupling}). For $t=2$, the claim follows directly from Lemma \ref{lem:uncoupling} \eqref{item:unlikely-uncoupling}, since
\begin{align*}
\P (E_2) = \P(F_1, E_2) 
&= \E[\I(F_1)\P(E_2 \mid X_1^{1:N}, \tilde X_1^{1:N})] 
\le \epsilon_N^{\mathrm{uc}}.
\end{align*}
For $t\ge 2$, we may partition the event $E_t$ with respect to the last time index for which the filter states were coupled: $E_t = \bigcup_{j=1}^{t-1} \{F_j, E_{j+1}, \ldots, E_t \}$.
As above, 
we have $\P(F_{t-1},E_t) \leq \epsilon_N^{\mathrm{uc}}$ and for $j\le t-2$ we may write
\begin{align*}
  \P(F_j,E_{j+1},\ldots,E_t) &= \E[\P(F_j,E_{j+1},\ldots,E_t \mid X_{j:j+1}^{1:N}, \tilde X_{j:j+1}^{1:N})] \\
  &= \E[\I(F_j,E_{j+1})\P(E_{j+2},\ldots,E_t \mid X_{j+1}^{1:N}, \tilde X_{j+1}^{1:N})] \\
  &\leq \P(F_j,E_{j+1}) \alpha^{\lfloor \frac{t-j-2}{s_N}\rfloor},
\end{align*}
where the inequality follows from Lemma \ref{lem:remain-uncoupled}. As above, Lemma \ref{lem:uncoupling} \eqref{item:unlikely-uncoupling} gives:
\begin{align*}
  \P(F_j,E_{j+1}) =\E[\P(F_j, E_{j+1} \mid X_j^{1:N}, \tilde X_j^{1:N})] 
  &\leq \epsilon_N^{\mathrm{uc}}.
\end{align*}
Therefore, 
\begin{align*}
  \P(E_t) 
  &\le \epsilon_N^{\mathrm{uc}} + \sum_{j=1}^{t-2} \P(F_j,E_{j+1},\ldots,E_t) \\
  & 
  \le \epsilon_N^{\mathrm{uc}} \bigg( 1 + \sum_{j=1}^{t-2} \alpha^{\lfloor \frac{t-j-2}{s_N}\rfloor} \bigg) \\
  & 
  \le \epsilon_N^{\mathrm{uc}} \bigg( 1 + s_N\sum_{j=0}^\infty \alpha^j\bigg),
\end{align*}
which is equal to the claimed upper bound.
\end{proof}

\begin{lemma}
  \label{lem:fwdcoupling-in-Hd}
  Let $(N_{\rm min}, c, \alpha)$ be from Lemma \ref{lem:tvuniform_ccpf}. For all $d\ge 0$, $t \in H^d$ and $N\ge N_{\rm min}$:
  $$
     \P(E_t) \le \epsilon_N^{E,d} \qquad \text{where} \qquad \epsilon_N^{E,d} = \epsilon_N^{E,0} \alpha^{\lfloor \frac{d}{s_N }\rfloor},
  $$
  where $s_N = \lceil c \log(N)\rceil + 1$ and the constant $\epsilon_N^{\mathrm{E,0}}$ is from Lemma \ref{lem:empty-unlikely}.
\end{lemma}
\begin{proof}
For $d=0$ the claim follows directly from Lemma \ref{lem:empty-unlikely}.
For $d\ge 1$, fix $t\in H^d$ and let $\ell = \sup\{j\in [T] \given j<t, x_j^*\neq \tilde{x}_j^*\}$ be the nearest index to the left of $t$ where the references are unequal. Note that if $F_j$ happens and $x_j^*=\tilde{x}_j^*$, then $F_{j+1}$ always happens by Lemma \ref{lem:uncoupling} \eqref{item:impossible-uncoupling}. Therefore, if $F_j$ happens for some $\ell \le j < t$, then $F_{j+1:t}$
happen. Furthermore, if no such finite $\ell$ exists, then because $F_1$ always happens, $\P(F_t)=1$, and the claim is trivial. Note also that $t-\ell \ge d$.  Therefore,
$$
   \P(E_t) = \P(E_\ell,E_{\ell+1},\ldots,E_t) = \P(E_{\ell+1},\ldots,E_t\mid E_\ell)\P(E_\ell),
$$
from which the claim follows by applying Lemmas \ref{lem:remain-uncoupled} and \ref{lem:empty-unlikely} to the terms on the right, respectively.
\end{proof}
\begin{lemma}
  \label{lem:sticky-forward-couplings}
For any $N\ge 1$, $t\in[T]$ and $1 \le k \le T-t$, if $F_t$ holds:
\begin{align*}
  \P((F_{t+1:t+k})^c \mid X_t^{1:N},\tilde{X}_t^{1:N}) \le k \epsilon_N^\mathrm{uc},
\end{align*}
where $\epsilon_N^\mathrm{uc}$ is from Lemma \ref{lem:fullcoupling-tv} in Appendix \ref{app:cpf-theory}.
\end{lemma}  
\begin{proof}
The event $(F_{t+1:t+k})^c$ can be partitioned with respect to the first time of uncoupling:
\[
(F_{t+1:t+k})^c = E_{t+1} \cup \left( \bigcup_{j=2}^k (F_{t+1:t+j-1} \cap E_{t+j}) \right),
\]
On the event $F_t$, it holds by Lemma \ref{lem:uncoupling} that
\[
\P(E_{t+1}\mid X_t^{1:N},\tilde{X}_t^{1:N}) \leq \epsilon_N^{\mathrm{uc}}.
\]
On the other hand,
\begin{align*}
&\P(F_{t+1:t+j-1},E_{t+j} \mid X_t^{1:N},\tilde{X}_t^{1:N}) \\
&= \E[\P(F_{t+1:t+j-1},E_{t+j} \mid X_{t:t+j-1}^{1:N},\tilde{X}_{t:t+j-1}^{1:N}) \mid X_t^{1:N},\tilde{X}_t^{1:N}] \\
&= \E[ \I(F_{t+1:t+j-2})\I(F_{t+j-1}) \P(E_{t+j} \mid X_{t+j-1}^{1:N},\tilde{X}_{t+j-1}^{1:N}) \mid X_t^{1:N},\tilde{X}_t^{1:N}] \\
&\leq \E[ \I(F_{t+1:t+j-2}) \epsilon_N^{\mathrm{uc}} \mid X_t^{1:N},\tilde{X}_t^{1:N} ] \\
&\leq \epsilon_N^{\mathrm{uc}},
\end{align*}
where we have applied Lemma \ref{lem:uncoupling} for the first inequality. 
The claim now follows by the union bound.
\end{proof}

\subsection{Backward index coupling}

The remaining results control how index coupled backward sampling (see lines \ref{line:backward-start}--\ref{line:backward-end} of Algorithm \ref{alg:mc_cbpf}) lead to coupling of output trajectories, relying on results above for the forward process.
Hereafter, we denote by $C_t = \{X_t^{J_t} =  \tilde{X}_t^{\tilde{J}_t}\}$ and $U_t = C_t^c$  the events that output states at time $t$ are coupled or are uncoupled, respectively. 
We also let $C_{T+1}$ stand for the trivial event which always happens.

Our first result gives bounds on recovering coupling in the backward pass, if $F_t$ happens. For that reason, let us denote the filtration of the forward process as $\G_{T+1} = \sigma(X_{1:T}^{1:N},\tilde{X}_{1:T}^{1:N})$, and its augmentation by the filtration of the backward indices up to time $t\in[T]$ as
$\G_t = \G_{T+1} \vee \sigma(J_u,\tilde{J}_u\given u\ge t)$. 

\begin{lemma}
  \label{lem:bs-coupling}
  For all $t\in[T]$, whenever $F_{t}$ holds:
  \begin{enumerate}[(i)~]
    \item \label{item:spontaneous-bs-couple} 
    $\P(C_t\mid \G_{t+1}) \ge \epsilon^\mathrm{bs}\frac{N}{N+1}$,
    \item \label{item:spontaneous-bs-couple-full} 
    $\P(C_t\mid \G_{t+1}) \ge \epsilon^\mathrm{bs}$, if also $x_t^* = \tilde{x}_t^*$,
    \item \label{item:bs-remain-coupled} 
    $\P(C_t\mid \G_{t+1}) \ge \frac{\epsilon^\mathrm{bs}N}{\epsilon^\mathrm{bs}N+1}$, if also $C_{t+1}$ holds,
    \item \label{item:bs-must-remain-coupled} 
    $\P(C_t\mid \G_{t+1}) = 1$ if both $C_{t+1}$ holds and $x_t^* = \tilde{x}_t^*$,
  \end{enumerate}
  where $\epsilon^\mathrm{bs} = \Mlo\Glo\Mhi^{-1}\Ghi^{-1}$.
\end{lemma}
\begin{proof}
Assume $t<T$. We start by noting that in Algorithm \ref{alg:mc_cbpf}, the indices $(J_t, \tilde J_t)$ are sampled from $\textsc{MaxCouple}\big(\mathrm{Categorical}(\omega^{0:N}), \mathrm{Categorical}(\tilde{\omega}^{0:N})\big)$,
where
\[
\omega^i = G_t(X_t^{i}) M_{t+1}(X_t^{i}, X_{t+1}^{J_{t+1}}), \quad \tilde \omega^i = G_t(\tilde{X}_t^{i}) M_{t+1}(\tilde{X}_t^{i},\tilde{X}_{t+1}^{\tilde{J}_{t+1}}).
\]
Note that $\Mlo\Glo  \le \tilde{\omega}^i, \omega^i \le \Mhi\Ghi$ and let $\epsilon^\mathrm{bs} = \Mlo\Glo\Mhi^{-1}\Ghi^{-1}$.
To establish \eqref{item:spontaneous-bs-couple}, we note that on the event $F_t$, $J_t = \tilde J_t \geq 1$ implies $X_t^{J_t} = \tilde X_t^{\tilde J_t}$. Thus,
\begin{align}
\P(C_t\mid \G_{t+1}) \geq \P(\cup_{i=1}^N \{J_t = \tilde J_t = i\} \mid \G_{t+1}) &= \sum_{i=1}^N \P(J_t = \tilde J_t = i \mid \G_{t+1}) \label{eq:spontanous-bs-couplings} \\
&
\geq N \frac{\epsilon^\mathrm{bs}}{N+1},\nonumber
\end{align}
where the last inequality follows from Lemma \ref{lem:ic-res} \eqref{item:ic-simple} in Appendix \ref{app:maxcouple}. If $t = T$, then
\begin{align}
\omega^i = G_T(X_T^{i}), \quad \tilde \omega^i = G_T(\tilde{X}_T^{i}), \label{eq:bigTomegas}
\end{align}
and as above, Lemma \ref{lem:ic-res} \eqref{item:ic-simple} implies that
\[
\P(C_T\mid \G_{T+1}) \geq N \frac{\Glo \Ghi^{-1}}{N+1} \geq N \frac{\epsilon^\mathrm{bs}}{N+1},
\]
establishing the claim \eqref{item:spontaneous-bs-couple}. The claim \eqref{item:spontaneous-bs-couple-full} is identical, except for that $J_t=\tilde{J}_t=0$ also implies $C_t$, so $i=0$ must be included in the bound \eqref{eq:spontanous-bs-couplings}.

Consider then \eqref{item:bs-remain-coupled}, assume $t < T$, and denote
\[
C = \{i \in\{1{:}N\}\given \omega^i = \tilde \omega^i\}.
\]
We observe that on the event $F_t \cap C_{t+1}$, $X_t^i = \tilde X_t^i$ and $\omega^i = \tilde \omega^i$ hold for all $i \geq 1$. Thus,
\[
\P(C_t\mid \G_{t+1}) \geq \P(J_t = \tilde J_t \in C \mid \G_{t+1} ) \geq \frac{|C| \epsilon^\mathrm{bs}}{|C|\epsilon^\mathrm{bs} + N+1-|C|} = \frac{N \epsilon^\mathrm{bs}}{N\epsilon^\mathrm{bs} +1},
\]
where the second inequality follows from Lemma \ref{lem:ic-res} \eqref{item:ic-coupled}. If $t=T$, we note that on the event $F_T$ both $X_T^i=\tilde X_T^i$ and $\omega^i = \tilde \omega^i$ hold for all $i\geq 1$, with $\omega^i, \tilde \omega^i$ given by \eqref{eq:bigTomegas}. Thus,
\[
\P(C_T\mid \G_{T+1}) \geq \frac{N \Glo \Ghi^{-1}}{N\Glo \Ghi^{-1} +1} \geq \frac{N \epsilon^\mathrm{bs}}{N\epsilon^\mathrm{bs} +1},
\]
establishing \eqref{item:bs-remain-coupled}. The claim \eqref{item:bs-must-remain-coupled}
follows by noting that if $x_t^* = \tilde x_t^*$,
then on the event $F_t \cap C_{t+1}$, $X_t^i = \tilde X_t^i$ and $\omega^i = \tilde \omega^i$ hold for all $i \geq 0$.
This means that by the coupling inequality $J_t = \tilde J_t$, and so $X_t^{J_t} = \tilde X_t^{\tilde J_t}$ with probability one.
\end{proof}

Our second result states that coupled output at $t$ is reasonably likely conditional on a sequence of successful forward couplings from $t$ to $t+k$.

\begin{lemma}
  \label{lem:backward-failing}
For all $N\ge 1$, $t\in[T]$ and $k\ge 1$:
$$
\P(U_t , F_{t:(t+k)\wedge T}) \le 
  \Big(1 - \epsilon^\mathrm{bs} \frac{N}{N+1}\Big)^{k+1} 
  + \frac{1}{\epsilon^\mathrm{bs} N + 1}  \Big(1 +  \frac{N+1}{N\epsilon^\mathrm{bs}}\Big).
$$
\end{lemma}
\begin{proof}
Assume first that $t+k < T$, and note that then, partitioning the event $U_t$ with respect to the last time index $C_{t+j}$ happened, we get
\begin{align}
\P(U_t , F_{t:t+k}) &= \P(U_{t:t+k}, F_{t:t+k}) + \sum_{j=1}^k \P(U_{t:t+j-1},C_{t+j}, F_{t:t+k}) \label{eq:split-index-couplings}, 
\end{align}
where we apply Lemma \ref{lem:bs-coupling} \eqref{item:spontaneous-bs-couple} for the first term:
\begin{align*}
\P(U_{t:t+k}, F_{t:t+k}) &= \E[ \P(U_{t:t+k}, F_{t:t+k} \mid \G_{t+1})] \\
&= \E[\I(F_t) \P(U_t \mid \G_{t+1})  \I(U_{t+1:t+k}, F_{t+1:t+k}) ] \\
&\leq \Big(1 - \epsilon^\mathrm{bs} \frac{N}{N+1}\Big) \P(U_{t+1:t+k}, F_{t+1:t+k}).
\end{align*}
Repeating the above calculation gives the bound
\[
\P(U_{t:t+k}, F_{t:t+k}) \leq \Big(1 - \epsilon^\mathrm{bs} \frac{N}{N+1}\Big)^{k+1}.
\]
The terms in the sum in \eqref{eq:split-index-couplings} may be upper bounded similarly by Lemma \ref{lem:bs-coupling} \eqref{item:spontaneous-bs-couple}: 
\begin{align*}
\P(U_{t:t+j-1}, C_{t+j}, F_{t:t+k}) &= \E[\P(U_{t:t+j-1}, C_{t+j}, F_{t:t+k} \mid \G_{t+1})] \\
&= \E[\I(F_t) \P(U_t \mid \G_{t+1}) \I(U_{t+1:t+j-1}, C_{t+j}, F_{t+1:t+k}) ] \\
&\leq \cdots \le \Big(1 - \epsilon^\mathrm{bs} \frac{N}{N+1}\Big)^{j-1} \P(U_{t+j-1}, C_{t+j}, F_{t+j-1:t+k}).
\end{align*}
Now apply Lemma \ref{lem:bs-coupling} \eqref{item:bs-remain-coupled}:
\begin{align*}
\P(U_{t+j-1}, C_{t+j}, F_{t+j-1:t+k}) &= \E[\I(F_{t+j:t+k}) \I(F_{t+j-1}, C_{t+j}) \P(U_{t+j-1} \mid \G_{t+j})] \\
&\leq  \frac{1 }{\epsilon^\mathrm{bs}N+1}.
\end{align*}

We conclude that 
\begin{align}
  \P(U_t, F_{t:t+k})
  &\le \Big(1 - \epsilon^\mathrm{bs} \frac{N}{N+1}\Big)^{k+1} + \frac{1}{\epsilon^\mathrm{bs} N + 1} \sum_{j=0}^\infty \Big(1 - \epsilon^\mathrm{bs} \frac{N}{N+1}\Big)^{j} \nonumber\\
  &\le \Big(1 - \epsilon^\mathrm{bs} \frac{N}{N+1}\Big)^{k+1} +   \frac{1}{\epsilon^\mathrm{bs} N + 1}\Big(\frac{N+1}{N\epsilon^\mathrm{bs}}\Big).\label{eq:sum-bound}
\end{align}

Assume now that $t+k\ge T$. Then, instead of \eqref{eq:split-index-couplings}, we get
\begin{align}
\P(U_t , F_{t:T}) &= \P(U_{t:T}, F_{t:T}) + \sum_{j=1}^{T-t} \P(U_{t:t+j-1},C_{t+j}, F_{t:T}), \label{eq:split-index-couplings2}
\end{align}
where an application of Lemma \ref{lem:bs-coupling} \eqref{item:bs-remain-coupled} gives
\[
\P(U_{t:T}, F_{t:T}) \leq \P(U_{T}, F_{T}) = \E[ \I(F_T, C_{T+1}) \P( U_T \mid \G_{T+1})] \leq \frac{1}{\epsilon^\mathrm{bs} N + 1},
\]
where $\G_{T+1} = \sigma(X_{1:T}^{1:N}, \tilde X_{1:T}^{1:N})$ and $\P(C_{T+1})=1$ by definition. The sum in \eqref{eq:split-index-couplings2} can be bounded by the latter term in 
\eqref{eq:sum-bound}. This leads to
\[
  \P(U_t , F_{t:T}) \le \frac{1}{\epsilon^\mathrm{bs} N + 1} + \frac{1}{\epsilon^\mathrm{bs} N + 1}\Big(\frac{N+1}{N\epsilon^\mathrm{bs}}\Big).
  \qedhere
\]
\end{proof}

The second result states that the probability of an uncoupled event decreases as the distance from $t$ to the index of the nearest uncoupled reference state increases.
\begin{lemma}
  \label{lem:uncoupling-main}
  Let $N_{\rm min}$ be as in Lemma \ref{lem:tvuniform_ccpf}.
For all $d\ge 0$, $t \in H^d$ and $N\ge N_{\rm min}$:
  $$
     \P(U_t) \le \epsilon_N^{E,d} + (1-\epsilon^\mathrm{bs})^{d} \epsilon^U_N,
  $$
  where $\epsilon_N^{E,d}$ is from Lemma \ref{lem:fwdcoupling-in-Hd}, 
  $\epsilon^\mathrm{bs} \in(0,1]$ is from Lemma \ref{lem:bs-coupling} (with the convention $0^0=1$) and 
  $$
  \epsilon^U_N \le c^U N^{-1/2}\log^2(N),
  $$
  where $c^U<\infty$ depends only on the constants in \textup{\ref{a:strong-mixing}}.
\end{lemma}
\begin{proof}
Fix $t \in H^d$. Note first that by Lemma \ref{lem:fwdcoupling-in-Hd}
$$
  \P(U_t)
  = \P(U_t,E_t) + \P(U_t,F_t)
  \le \epsilon_N^{E,d} + \P(U_t,F_t).
$$
Let then $u = \inf\{ j\ge t \given  x_j^*\neq\tilde{x}_j^*\}$ be the nearest index to the right where the references differ. If no such finite $u$ exists, then $x_{t:T}^*=\tilde{x}_{t:T}^*$, and so $F_t=F_{t:T}$ by Lemma \ref{lem:uncoupling} \eqref{item:impossible-uncoupling} and $\{U_t,F_{t:T}\} = \{U_{t:T},F_{t:T}\}$ by Lemma \ref{lem:bs-coupling} \eqref{item:bs-must-remain-coupled}. Because $\P(C_T\mid F_T)=1$, we have $\P(U_t,F_t)=0$.

Otherwise, $u-t \ge d$, and by Lemma \ref{lem:uncoupling} \eqref{item:impossible-uncoupling} $F_t = F_{t:u}$ and by Lemma \ref{lem:bs-coupling} \eqref{item:bs-must-remain-coupled} $\{U_t,F_{t:u}\} = \{U_{t:u}, F_{t:u}\}$. Therefore, if $d\ge 1$, 
\begin{align*}
  \P(U_t, F_t) & = \P(U_{t:u}, F_{t:u})  \\
  &= \E[ \P(U_t\mid \G_{t+1}) \I(F_t) \I(U_{t+1:u}, F_{t+1:u})] \\
  &\le (1 - \epsilon^\mathrm{bs} ) \P(U_{t+1:u}, F_{t+1:u}) \\
  &\le (1 - \epsilon^\mathrm{bs} )^{u-t}\P(U_u,F_u),
\end{align*}
by Lemma \ref{lem:bs-coupling} \eqref{item:spontaneous-bs-couple-full} and recursion.
This bound is valid also for $d=0$, that is, $u=t$. For $k\ge 1$, we may write
\begin{align*}
  \P(U_u, F_{u}) 
  &= \P(U_u, F_{u:u+k}) 
  + \P(U_u, F_{u}, (F_{u+1:u+k})^c )
\end{align*}
where $F_{j}$ for $j>T$ stand for the trivial events $\P(F_j)=1$.
For the latter term, we use Lemma \ref{lem:sticky-forward-couplings}, and 
for the former, we use Lemma \ref{lem:backward-failing} to deduce that
\begin{equation}
\epsilon_U^N = 
\P(U_u, F_{u}) 
\le  \Big(1 - \epsilon^\mathrm{bs} \frac{N}{N+1}\Big)^{k+1} 
+ \frac{1}{\epsilon^\mathrm{bs} N + 1}  \Big(1 +  \frac{N+1}{N\epsilon^\mathrm{bs}}\Big) 
+ k \epsilon_N^\mathrm{uc}.
\label{eq:epsilon-U-bound}
\end{equation}
Note that $\epsilon_N^\mathrm{uc} \le c^\mathrm{uc} N^{-1/2}$ where $c^\mathrm{uc}$ depends on the constants in \ref{a:strong-mixing}. Choosing $k = \lceil \log^2(N+1) \rceil$ ensures that the first term of the upper bound \eqref{eq:epsilon-U-bound} is no more than
$$
\gamma^{\log^2(N+1)} 
= (N+1)^{\log\gamma \log (N+1)} \quad \text{where}\quad \gamma = \Big(1 - \epsilon^\mathrm{bs} \frac{N_{\mathrm{min}}}{N_{\mathrm{min}}+1}\Big) \in (0,1),
$$
which is $O(N^{-1})$. The latter two terms in \eqref{eq:epsilon-U-bound} are $O(N^{-1})$ and $O(N^{-1/2}\log^2(N))$, respectively, and therefore the claimed constant $c^U<\infty$ exists.
\end{proof}

\subsection{Proof of Theorem \ref{thm:main-contraction}}
\label{app:proof-main-contraction}

Let $(N_{\rm min}, c, \alpha)$ be from Lemma \ref{lem:tvuniform_ccpf}.
Lemma \ref{lem:number-to-individual-sites} and Lemma \ref{lem:uncoupling-main} imply that 
whenever $N\ge N_{\rm min}$, we have the upper bound $\E[B] \le \lambda_N b^*$, where
\begin{align*}
\lambda_N &= \epsilon_N^{E,0} + \epsilon^U_N + 
2 \sum_{d=1}^\infty \big[\epsilon_N^{E,d} + (1-\epsilon^\mathrm{bs})^{d} \epsilon^U_N\big],
\end{align*}
By Lemma \ref{lem:fwdcoupling-in-Hd}, 
$\epsilon_N^{E,d} = \epsilon_N^{E,0} \alpha^{\lfloor \frac{d}{s_N }\rfloor}$ where $s_N = \lceil c \log(N)\rceil + 1$, so
\begin{align*}
\lambda_N &= \epsilon_N^{E,0} \bigg( 1 + 2\sum_{d=1}^\infty \alpha^{\lfloor \frac{d}{s_N }\rfloor}\bigg) 
+ \epsilon^U_N \bigg( 1 + 2\sum_{d=1}^\infty (1-\epsilon^\mathrm{bs})^{d}\bigg) \\
&\le \epsilon_N^{E,0}\bigg( 1 + 2\frac{s_N}{1-\alpha} \bigg) + 2\frac{\epsilon^U_N}{\epsilon^\mathrm{bs}}.
\end{align*}
Recall from Lemmas \ref{lem:empty-unlikely} and \ref{lem:uncoupling} that 
$$
\epsilon_N^{E,0} = \frac{\Mhi\Ghi}{2\Mlo\Glo} \frac{1}{\sqrt{N+1}} \Big( 1 + \frac{ s_N }{1-\alpha} \Big) \le c^{E,0} \frac{ \log N}{\sqrt{N}},
$$
where $c^{E,0}$ only depends on the constants in \ref{a:strong-mixing}.
Because $\epsilon^U_N \le c^U N^{-1/2}\log^2 N$ by Lemma \ref{lem:uncoupling-main}, we deduce an upper bound $\lambda_N \le c_\lambda N^{-1/2}\log^2 N$.
\hfill$\qed$

\section{Maximal coupling algorithms}
\label{app:maxcouple}

Algorithm \ref{alg:maxcouple-generic} is a generic (rejection sampler algorithm) which samples from a maximal coupling of distributions $p$ and $q$ (having densities $p$ and $q$ with respect to any common $\sigma$-finite dominating measure).
\begin{algorithm}[H]
  \caption{\textsc{MaxCouple}$\big(p, q)$}
  \label{alg:maxcouple-generic} 
  \begin{algorithmic}[1]
  \State Draw $X\sim p$
  \State \textbf{with probability} $1\wedge \frac{q(X)}{p(X)}$ 
  \textbf{output} $(X,X)$
  \Loop
  \State Draw $Y \sim q$
  \State \textbf{with probability} $1 - \big(1 \wedge \frac{p(Y)}{q(Y)}\big)$ 
  \textbf{output} $(X,Y)$
  \EndLoop
  \end{algorithmic}
\end{algorithm}

For discrete distributions, maximal coupling can be implemented directly. 
Algorithm \ref{alg:maxcouple-categoricals} describes how categorical distributions can be sampled from. 
\begin{algorithm}[H]
  \caption{\textsc{MaxCouple}$\big(\mathrm{Categorical}(\omega^{0:N}),\mathrm{Categorical}(\tilde{\omega}^{0:N})\big)$}
  \label{alg:maxcouple-categoricals} 
  \begin{algorithmic}[1]
  \State Calculate $w^{0:N} \gets \frac{\omega^{0:N}}{\sum_{j=0}^N
    \omega^{j}}$;
  $\tilde{w}^{0:N} \gets \frac{\tilde{\omega}^{0:N}}{\sum_{j=0}^N
  \tilde{\omega}^{j}}$ and draw $U \sim U(0,1)$
  \If{$U\le  
  \sum_{j=0}^N w^{j}\wedge \tilde{w}^{j}$} 
  \State 
  $\tilde{I}\gets I \sim \mathrm{Categorical}(w^{0:N}\wedge \tilde{w}^{0:N})$
  \Else
  \State 
  $I \sim \mathrm{Categorical}(w^{0:N} - w^{0:N}\wedge \tilde{w}^{0:N})$; $\tilde{I} \sim \mathrm{Categorical}(\tilde{w}^{0:N} - w^{0:N}\wedge \tilde{w}^{0:N})$
  \EndIf
  \State \textbf{output} ($I,\tilde{I}$)
  \end{algorithmic}
\end{algorithm}

The following result, which is a slight modification of \citep[Lemma 13]{lee-singh-vihola}, gives lower bounds for the coupling probabilities when using Algorithm \ref{alg:maxcouple-categoricals}.

\begin{lemma}
  \label{lem:ic-res} 
  Suppose $0<\omega_*\le \omega^*<\infty$ and
  $\omega^{i},\tilde{\omega}^{i}\in[\omega_*,\omega^*]$
  for $i=0{:}N$. Let
  \[
  \varepsilon = \frac{\omega_*}{\omega^*}
  \qquad\text{and}\qquad
  C \subset \{j\in\{0{:}N\}\given
            \omega^{j}=\tilde{\omega}^{j}\}.
  \]
  Then, 
  $$
  (I,\tilde{I}) \sim \textsc{MaxCouple}\big(\mathrm{Categorical}(\omega^{0:N}), \mathrm{Categorical}(\tilde{\omega}^{0:N})\big)
  $$
  satisfy the following:
  \begin{enumerate}[(i)~]
      \item \label{item:ic-simple}
        $\P(I=\tilde{I}=i) \ge
        \frac{\varepsilon}{N+1}$
        for all $i=0{:}N$,
      \item \label{item:ic-coupled}
        $\P(I=\tilde{I}\in C)
      \ge \frac{|C| \varepsilon}{|C|\varepsilon + N+1-|C|}$.
  \end{enumerate}
\end{lemma}
\begin{proof} 
    Note that $\P(I=\tilde{I}=i) = w^{i}\wedge
    \tilde{w}^{i}$, where $w^i = \omega^i / \sum_{j=0}^N \omega_j$ and 
    $\tilde{w}^i = \tilde{\omega}^i / \sum_{j=0}^N \tilde{\omega}_j$ (cf.~Algorithm \ref{alg:maxcouple-categoricals} in Appendix \ref{app:maxcouple}), so the first bound is immediate.
    For the second, let $C^c =
    \{0{:}N\}\setminus C$, and observe that
    \begin{align*}
        \sum_{j\in C} w^{j} \wedge \tilde{w}^{j}
        &= \frac{\sum_{j\in C} \omega^{j}}{\sum_{i\in C} \omega^{i} +
          \big(\sum_{i\in C^c} \omega^{i}\big) \vee
          \big(\sum_{i\in C^c} \tilde{\omega}^{i}\big)}
        \\
        &\ge \frac{ |C| \omega_*
           }{
          |C| \omega_*  + |C^c| \omega^*},
    \end{align*}
    because $x\mapsto x(x+b)^{-1}$ is increasing for $x\ge 0$ for any $b>0$.
    The last bound equals \eqref{item:ic-coupled}.
\end{proof}

\section{Properties of CPF and CBPF transitions}
\label{app:cpf-cbpf}

Algorithm \ref{alg:cpf} summarises the original conditional particle filter introduced in \citep{andrieu-doucet-holenstein}. Note that marginalising $A_{t-1}^i$, lines \ref{line:ancestor-indices}--\ref{line:sampling-particles} of Algorithm \ref{alg:cpf} coincide with sampling from the mixture in line \ref{line:mixture-transition} of Algorithm \ref{alg:cbpf}. The only real difference is line \ref{line:ancestor-trace} of Algorithm \ref{alg:cpf} which implements `ancestor tracing' (AT) instead of backward sampling (BS) implemented in lines \ref{line:cbpf-backward-sample1}--\ref{line:cbpf-backward-sample2} of Algorithm \ref{alg:cbpf}.
\begin{algorithm}
  \caption{\textsc{CPF}($x_{1:T}^*,N$).}
  \label{alg:cpf}
  \begin{algorithmic}[1]
      \State $X_{1:T}^{0} \gets x_{1:T}^*$ 
      \State $X_1^{i}\sim M_1(\uarg)$ \Comment{for $i\in\{1{:}N\}$}
      \State $W_1^{i} \gets G_1(X_1^{i})$ \Comment{for $i\in\{0{:}N\}$}
      \For{$t=2,\ldots,T$}
      \State $A_{t-1}^{i} \sim \mathrm{Categorical}(W_{t-1}^{0:N})$ \label{line:ancestor-indices} \Comment{for $i\in\{1{:}N\}$}
      \State 
      $X_t^{i} \sim M_t(X_{t-1}^{A_{t-1}^i}, \uarg)$ \label{line:sampling-particles} \Comment{for $i\in\{1{:}N\}$}
      \State $W_t^{i} \gets G_t(X_t^{i})$  \Comment{for $i\in \{0{:}N\}$} 
      \EndFor
      \State $J_T^*  \sim \mathrm{Categorical}(W_T^{0:N})$
      \For{$t=T-1,T-2,\ldots,1$} 
      \State $J_t^* = A_t^{J_{t+1}^*} 1(J_{t+1}^* \neq 0)$ \label{line:ancestor-trace}
      \EndFor
      \State \textbf{output}
      $(X_1^{J_1^*},\ldots,X_T^{J_T^*})$
  \end{algorithmic}
\end{algorithm}

Let $P_N^\mathrm{BS}(x_{1:T}^*,\uarg)$ and $P_N^\mathrm{AT}(x_{1:T}^*,\uarg)$ stand for the Markov transition defined by Algorithms \ref{alg:cbpf} and \ref{alg:cpf}, respectively. Both $P_N^\mathrm{BS}$ and $P_N^\mathrm{AT}$ are reversible with respect to $\pi$ \citep[Proposition 9]{chopin-singh}, define positive operators in $L^2$, and are ordered with respect to lag-1 stationary autocorrelations \citep[Theorem 10]{chopin-singh}.

Let then $\check{P}_N^\mathrm{BS}(x_{1:T}^*; j_{1:T}, \uarg )$ and $\check{P}_N^\mathrm{AT}(x_{1:T}^*; j_{1:T}, \uarg )$ stand for the augmentations of $P_N^\mathrm{BS}$ and $P_N^\mathrm{AT}$ which include also the laws of $J_{1:T}$ and $J_{1:T}^*$, respectively; that is, with $\P$ standing for the randomness in Algorithm \ref{alg:cpf},
$$
\check{P}_N^\mathrm{AT}(x_{1:T}^*; j_{1:T}, A) 
= \P\big((X_1^{J_1^*},\ldots,X_T^{J_T^*})\in A, J_{1:T}^*=j_{1:T}\big),
$$
and similarly for $\check{P}_N^\mathrm{BS}$. Clearly, marginalising over $j_{1:T}$ yields the original transitions for any $x_{1:T}^*\in\X^T$:
$$
    \sum_{j_{1:T}\in \{0{:}N\}^T} \check{P}_N^\mathrm{AT}(x_{1:T}^*; j_{1:T}, A) = P_N^\mathrm{AT}(x_{1:T}^*, A),
$$
and similarly for $\check{P}_N^\mathrm{BS}$. 

Let then $F = \big\{j_{1:T}\in \{0{:}N\}^T\given j_1\neq 0,\ldots,j_T\neq 0\}$ stand for the events where none of the reference states is selected. In such an event, the algorithms coincide:
\begin{lemma}
  For all $x_{1:T}^*\in\X^T$, $j_{1:T} \in F$ and measurable $A\subset \X^T$:
\begin{equation}
   \check{P}_N^\mathrm{AT}(x_{1:T}^*; j_{1:T}, A) = 
   \check{P}_N^\mathrm{BS}(x_{1:T}^*; j_{1:T}, A).
   \label{eq:cpf-cbpf-correspondence}
\end{equation}
\end{lemma}
\begin{proof}
We may write the joint density of $X_{1:T}^{1:N}$ and $A_{1:T-1}^{1:N}$ (the latter wrt. counting measure on $\{0{:}N\}$) as follows:
$$
\mu(x_{1:T}^{1:N}, a_{1:T-1}^{1:N}) = 
M_{1}^{\otimes N}(x_{1}^{1:N})
\prod_{t=2}^{T-1}\left(\prod_{i=1}^{N}\frac{G_{t-1}(x_{t-1}^{a_{t-1}^{i}})}{\sum_{j=0}^{N}G_{t-1}(x_{t-1}^{j})}M_{t}(x_{t-1}^{a_{t-1}^{i}},x_{t}^{i})\right),
$$
and the joint density of $X_{1:T}^{1:N}$, $A_{1:T-1}^{1:N}$ and $J_{1:T}^*$ is
$$
\mu(x_{1:T}^{1:N}, a_{1:T-1}^{1:N}) 
\bigg( \frac{G_{T}(x_{T}^{j_{T}^{\ast}})}{\sum_{j=0}^{N}G_{T}(x_{T}^{j})} \bigg)
\prod_{t=1}^{T-1} 
a_t^{j_{t+1}^*} 1(j_{t+1}^* \neq 0).
$$
If $j_{1:T}^*\in F$, then the indicator is always one, and then marginalising $a_{T-1}^{1:N}$, $a_{T-2}^{1:N}$, \ldots, $a_1^{1:N}$ leads to 
the probability $\check{P}_N^\mathrm{AT}(x_{1:T}^*; j_{1:T}^*, A)$ being:
$$
\int 
\Bigg[
  \prod_{i=1}^{N} 
M_1(x_1^i)
\Bigg]
\Bigg[
  \prod_{t=2}^T 
  C_t(j_{t-1:t}^*, x_{t-1:t}^{1:N})
\Bigg] 
\bigg( \frac{G_{T}(x_{T}^{j_{T}^*})}{\sum_{j=0}^{N}G_{T}(x_{T}^{j})} \bigg) 
1(x_{1:T}^{j_{1:T}^*} \in A),
$$
where the product terms may be written as follows:
\begin{align*}
C_t(j_{t-1:t}^*, x_{t-1:t}^{1:N}) 
&= \frac{G_{t-1}(x_{t-1}^{j_{t-1}^*})}{\sum_{j=0}^{N}G_{t-1}(x_{t-1}^{j})}M_{t}(x_{t-1}^{j_{t-1}^*},x_{t}^{j_{t}^{\ast}})
\prod_{\substack{i=1 \\ i\neq j_{t}^{\ast}}}^{N}\Phi_{t}^{x_{t-1}^{\ast}}\left(\frac{1}{N}\sum_{i=1}^{N}\delta_{x_{t-1}^{i}}\right)(x_{t}^{i}) \\
&=\frac{G_{t-1}(x_{t-1}^{j_{t-1}^*})M_{t}(x_{t-1}^{j_{t-1}^*},x_{t}^{j_{t}^{\ast}})}{\sum_{l=0}^{N}G_{t-1}(x_{t-1}^{l})M_{t}(x_{t-1}^{l},x_{t}^{j_{t}^{\ast}})}\prod_{i=1}^{N}\Phi_{t}^{x_{t-1}^{\ast}}\left(\frac{1}{N}\sum_{i=1}^{N}\delta_{x_{t-1}^{i}}\right)(x_{t}^{i}).
\end{align*}
We conclude that
\begin{align*}
\check{P}_N^\mathrm{AT}(x_{1:T}^*; j_{1:T}^*, A) &= \int 
\Bigg[M_{1}^{\otimes N}(x_{1}^{1:N})\prod_{t=2}^{T-1}\mathbf{M}_{t}^{x_{t-1}^{\ast}}(x_{t-1}^{1:N},x_{t}^{1:N})\Bigg] \times \\
&\Bigg[\prod_{t=2}^{T-1}\frac{G_{t-1}(x_{t-1}^{j_{t-1}^{\ast}})M_{t}(x_{t-1}^{j_{t-1}^{\ast}},x_{t}^{j_{t}^{\ast}})}{\sum_{l=0}^{N}G_{t-1}(x_{t-1}^{l})M_{t}(x_{t-1}^{l},x_{t}^{j_{t}^{\ast}})}\Bigg]
\frac{G_{T}(x_{T}^{j_{T}^{\ast}})}{\sum_{j=0}^{N}G_{T}(x_{T}^{j})}
1(x_{1:T}^{j_{1:T}^*} \in A).
\end{align*}
It is direct to verify that this coincides with $\check{P}_N^\mathrm{BS}(x_{1:T}^*; j_{1:T}^*, A)$.
\end{proof}
Now, \citep[Corollary 14]{andrieu-lee-vihola} implies that (the result is stated for $P_N^\mathrm{AT}$, but proofs of \citep[Propositions 6 \& 13]{andrieu-lee-vihola} reveal that the bound below holds):
\begin{equation}
\sum_{j_{1:T}\in F} \check{P}_N^\mathrm{AT}(x_{1:T}^*; j_{1:T}, A) \ge \epsilon_{N,T} \pi(A),
\label{eq:minorisation}
\end{equation}
where
$$
\epsilon_{N,T} = \bigg( \frac{1 - \frac{1}{N+1}}{1 + \frac{2(\alpha-1)}{N+1} }\bigg)^T = 
\bigg( 1 + \frac{2\alpha -1}{N}\bigg)^{-T} \ge \exp\bigg( -T\frac{2\alpha - 1}{N}\bigg),
$$
where $\alpha\in(0,\infty)$ depends only on \ref{a:strong-mixing} and the last inequality follows by $\log(1+x) \le x$ which holds for all $x\ge 0$.

Thanks to \eqref{eq:cpf-cbpf-correspondence} and \eqref{eq:minorisation}, we have that:
$$
   \tv{ P_N^\mathrm{BS}(x_{1:T}^*, \uarg) - \pi } \le 1 - \epsilon_{N,T} \le 1 - e^{-T C/N},
$$
where $C=2\alpha - 1$.

\section{Coupling CBPFs with pairwise potentials}
\label{app:pairpotcoup}

The CBPF (Algorithm \ref{alg:cbpf}) can, be generalised to pairwise potentials of the form $G_t(x_{t-1},x_t)$ in a number of ways. Two approaches in the literature \citep{lee-singh-vihola,lindsten-jordan-schon} rely on sampling ancestor indices, and their couplings are therefore based on index couplings.  The key property behind the couplings IMC and JMC discussed in this paper is to avoid such index coupling entirely.

It therefore seems desirable to seek for an extension of Algorithm \ref{alg:cbpf} for pairwise potentials, avoiding index coupling. This can be achieved by modifying Algorithm \ref{alg:cbpf} slightly, replacing 
lines \ref{line:cbpf-forward-end} and \ref{line:cbpf-backward-sample1} with the following:
\begin{align}
W_t^{i} &\gets \sum_{k=0}^N \frac{W_{t-1}^k M_t(X_{t-1}^k, X_t^i ) G_t(X_{t-1}^k, X_t^{i})}{\sum_{\ell=0}^N W_{t-1}^\ell M_t(X_{t-1}^\ell, X_t^i )} , & \text{for }i \in \{0{:}N\} \label{eq:marginal-weights}\\
B_t^{i} &\gets W_t^{i} M_{t+1}(X_t^{i}, X_{t+1}^{J_{t+1}})G_{t+1}(X_t^i,X_{t+1}^{J_{t+1}}), 
& \text{for }i \in \{0{:}N\},\label{eq:marginal-backward-weights}
\end{align}
respectively. This algorithm is a conditional version of the marginal particle filter \citep{klaas-defreitas-doucet}, and we call this modification of Algorithm \ref{alg:cbpf} as $\textsc{MarginalCBPF}(x_{1:T}^*,N)$.

Our coupling Algorithm \ref{alg:mc_cbpf} is direct to adapt to this generalisation: we need to simply use \ref{eq:marginal-weights} and \ref{eq:marginal-backward-weights} in place of weight updates in lines \ref{line:coupled-weights} and \ref{line:coupled-backward-weights} of Algorithm \ref{alg:mc_cbpf}, respectively. The computational complexity of \ref{eq:marginal-weights} is $O(N^2)$ similar to IMC and JMC.

$\textsc{MarginalCBPF}$ has not appeared in the literature before, and it might not be clear that it is still defines a valid Markov update like CBPF. We therefore conclude with a proof that it is indeed reversible with respect to the desired target distribution:
\begin{proposition}
  \sloppy
  $\mathbf{S}_k \gets \textsc{MarginalCBPF}(\mathbf{S}_{k-1}, N)$
  with any $N\ge 1$ defines a Markov transition reversible with respect to probability density $\pi(x_{1:T}) = Z^{-1} M_{1}(x_{1})G_{1}(x_{1})\prod_{t=2}^{T}G_{t}(x_{t-1},x_{t})M_{t}(x_{t-1},x_{t})$, where $Z>0$ is a normalising constant.
\end{proposition}
\begin{proof}
Following the approach of \cite{andrieu-doucet-holenstein}, let us write the extended distribution 
$$
  \bar{\pi}(j_{1:T},x_{1:T}^{0:N})=\frac{1}{(N+1)^{T}}\pi(x_{1:T}^{j_{1:T}})\prod_{i\neq j_{1}}M_{1}(x_{1}^{i})\prod_{t=2}^{T}\prod_{\substack{i=0 \\i\neq j_{t}}}^{N}\frac{\sum_{k=1}^{N}W_{t-1}^{k}M_{t}(x_{t-1}^{k},x_{t}^{i})}{\sum_{\ell=0}^{N}W_{t-1}^{\ell}},
$$
which is the density of sampling $J_{1:T}$ uniformly on $\{0{:}N\}^{T}$, setting $X_{1:T}^{J_{1:T}} = X_{1:T}^{*}\sim\pi$ and then running 
lines \ref{line:cbpf-forward-start}--\ref{line:cbpf-forward-end} of Algorithm \ref{alg:cbpf}, with line \ref{line:cbpf-forward-end} replaced by \eqref{eq:marginal-weights}, and with $i\in \{0{:}N\}\setminus j_1$ in line \ref{line:cbpf-forward-start} and $i\in \{0{:}N\}\setminus j_t$ in \eqref{eq:marginal-weights}.

We will show that the selection step of \textsc{MarginalCBPF}, given particles $x_{1:T}^{0:N}$, corresponds to sampling $J_{1:T}\sim\bar{\pi}(\cdot\mid x_{1:T}^{1:N})$. We write $W_t^i = U_{t}^{i}/D_{t}^{i}$ in terms of the numerator and the denominator in \eqref{eq:marginal-weights},
and similarly let $W_{1}^{i}=U_{1}^{i}/D_{1}^{i}$, where $U_{1}^{i}=M_{1}(x_{1}^{i})G_{1}(x_{1}^{i})$ and $D_{1}^{i}=M_{1}(x_{1}^{i})$.
  
We have $B_{T}^{i}=W_{T}^{i}$ and given the selected index $j_{t+1}$,
$$
  B_{t}^{i}=W_{t}^{i}M_{t+1}(x_{t}^{i},x_{t+1}^{j_{t+1}})G_{t+1}(x_{t}^{i},x_{t+1}^{j_{t+1}}),
$$
so that (omitting the bounds of sum $k=0,\ldots,N$ for brevity):
$$
  \frac{B_{t}^{j_{t}}}{\sum_{k}B_{t}^{k}}
  =\frac{W_{t}^{j_{t}}M_{t+1}(x_{t}^{j_{t}},x_{t+1}^{j_{t+1}})G_{t+1}(x_{t}^{j_{t}},x_{t+1}^{j_{t+1}})}{
    \sum_{k}W_{t}^{k}M_{t+1}(x_{t}^{k},x_{t+1}^{j_{t+1}})G_{t+1}(x_{t}^{k},x_{t+1}^{j_{t+1}})}
$$
where we observe that the denominator on the right equals $U_{t+1}^{j_{t+1}}$. Hence, we have
\begin{align*}
  \frac{B_{T}^{j_{T}}}{\sum_{k}B_{T}^{k}}\prod_{t=1}^{T-1}\frac{B_{t}^{j_{t}}}{\sum_{k}B_{t}^{k}} 
  & = \frac{1}{\sum_{k}W_{T}^{k}}\frac{U_{T}^{j_{T}}}{D_{T}^{j_{T}}}
    \prod_{t=1}^{T-1}\frac{U_{t}^{j_{t}}}{D_{t}^{j_{t}}}\frac{M_{t+1}(x_{t}^{j_{t}},x_{t+1}^{j_{t+1}})G_{t+1}(x_{t}^{j_{t}},x_{t+1}^{j_{t+1}})}{U_{t+1}^{j_{t+1}}}.
\end{align*}
The $U^{j_t}_t$ terms cancel, except for $U^{j_1}_1 = M_{1}(x_{1}^{j_{1}})G_{1}(x_{1}^{j_{1}})$, so we conclude that
\begin{align*}
  \frac{B_{T}^{j_{T}}}{\sum_{k}B_{T}^{k}}\prod_{t=1}^{T-1}\frac{B_{t}^{j_{t}}}{\sum_{k}B_{t}^{k}} 
   & =\frac{1}{\sum_{k}W_{T}^{k}}\frac{M_{1}(x_{1}^{j_{1}})G_{1}(x_{1}^{j_{1}})
   \prod_{t=2}^{T}M_{t}(x_{t-1}^{j_{t-1}},x_{t}^{j_{t}})G_{t}(x_{t-1}^{j_{t-1}},x_{t}^{j_{t}})}{\prod_{t=1}^{T}D_{t}^{j_{t}}}.
\end{align*}
Let us denote the selection probability above by $S(j_{1:T}\mid x_{1:T}^{0:N})$, then we find
\begin{align*}
  \frac{\bar{\pi}(j_{1:T},x_{1:T}^{0:N})}{S(j_{1:T}\mid x_{1:T}^{0:N})} 
   & =\frac{\frac{1}{(N+1)^{T}}\frac{1}{Z}\left\{ \prod_{t=1}^{T}\prod_{i\neq j_{t}}\frac{\sum_{k}W_{t-1}^{k}M_{t}(x_{t-1}^{k},x_{t}^{i})}{\sum_{\ell}W_{t-1}^{\ell}}\right\} }{\frac{1}{\sum_{k}W_{T}^{k}}\frac{1}{\prod_{t=1}^{T}D_{t}^{j_{t}}}}\\
   & =\frac{1}{Z}\prod_{t=1}^{T}\left\{ \frac{1}{N}\sum_{k}W_{t}^{k}\right\} \prod_{i}\frac{\sum_{k}W_{t-1}^{k}M_{t}(x_{t-1}^{k},x_{t}^{i})}{\sum_{\ell} W_{t-1}^{\ell}},
  \end{align*}
  where we have used the definition of $D_t^{j_t}$.
  Since the ratio is a constant (i.e.~does not depend on $j_{1:T}$),
  we may conclude that $S(j_{1:T}\mid x_{1:T}^{0:N})=\bar{\pi}(j_{1:T}\mid x_{1:T}^{0:N})$.
  That is, the selection of the indices backwards corresponds to sampling
  from the conditional distribution for $J_{1:T}$ given the particle
  system $X_{1:T}^{1:N}$.  
\end{proof}

\section{Further experimental results}
\label{app:extraexperiments}

Figure \ref{fig:lgIterations} shows the mean number of iterations until coupling for the linear-Gaussian model.
\begin{figure}
  \includegraphics[scale=0.9]{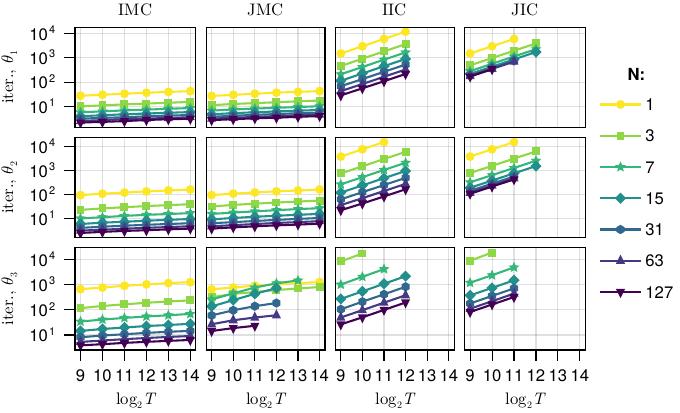}
\caption{Average coupling times in Algorithm \ref{alg:unbiased} (using IMC, JMC, IIC and JIC) for the linear-Gaussian model (Section \ref{sec:lgmodel}) with parameters $\theta_1$, $\theta_2$ and $\theta_3$. The experiments that failed to complete within the time limit of 8 hours are omitted from the graphs.}
  \label{fig:lgIterations}
\end{figure}

\end{document}